\documentclass[journal]{IEEEtran}
\usepackage{fix-cm}
\usepackage{cite}
\usepackage{amsmath,amssymb,amsfonts}
\usepackage{algorithmic}
\usepackage{graphicx,color}
\usepackage{textcomp}
\usepackage{tabularx}
\usepackage{amssymb}
\usepackage[utf8]{inputenc}
\usepackage{color}
\usepackage{authblk}
\usepackage{amsthm}
\newtheorem{theorem}{Theorem}
\usepackage{xcolor}
\usepackage{array, balance}
\usepackage{graphicx}
\usepackage{subcaption}
\usepackage{algorithmic,algorithm}
\usepackage{amsmath}
\usepackage[capitalise]{cleveref}
\usepackage{booktabs}
\usepackage{float}
\usepackage{multirow}
\usepackage{lettrine}
\usepackage{colortbl}
\usepackage{hhline}
\usepackage{nccmath}
\newcommand{\tref}[1]{TABLE \ref{#1}}

\newcommand{\fref}[1]{Figure \ref{#1}}
\NewDocumentCommand{\smallmath}{O{0.85} m}{%
  \scalebox{#1}{$\displaystyle #2$}%
}
\Crefname{figure}{Fig.}{Figs.}

\newcommand{\updated}[1]{{#1}}

\def\BibTeX{{\rm B\kern-.05em{\sc i\kern-.025em b}\kern-.08em
    T\kern-.1667em\lower.7ex\hbox{E}\kern-.125emX}}
\AtBeginDocument{\definecolor{ojcolor}{cmyk}{0.93,0.59,0.15,0.02}}

\begin{document}
\newcommand{\mytitle}{CIVIC: Cooperative Immersion Via Intelligent Credit-sharing in DRL-Powered Metaverse}
\title{\mytitle}


\author{Amr Aboeleneen\IEEEauthorrefmark{1}\IEEEmembership{(Member, IEEE)}, Mohamed Abdallah\IEEEauthorrefmark{1}\IEEEmembership{(Senior, IEEE)}, Aiman Erbad\IEEEauthorrefmark{2}\IEEEmembership{(Senior, IEEE)} and Amr Salem\IEEEauthorrefmark{2}\IEEEmembership{(Senior, IEEE)}
\thanks {This work was made possible by the GSRA grant \# GSRA9-L-1-0518-22022 from the Qatar National Research Fund (a member of Qatar Foundation). The findings achieved herein are solely the responsibility of the authors.\\ Amr Aboeleneen and Mohamed Abdallah are with the College of Science and Engineering, Hamad Bin Khalifa University, Qatar.\\ Aiman Erbad and Amr Salem are with the College of Engineering, Qatar University, Qatar.  }
}


\markboth{\mytitle}{A.Aboeleneen \textit{et al.}}
\maketitle

\begin{abstract}
The Metaverse faces complex resource allocation challenges due to diverse Virtual Environments (VEs), Digital Twins (DTs), dynamic user demands, and strict immersion needs. This paper introduces CIVIC (Cooperative Immersion Via Intelligent Credit-sharing), a novel framework optimizing resource sharing among multiple Metaverse Service Providers (MSPs) to enhance user immersion. Unlike existing methods, CIVIC integrates VE rendering, DT synchronization, credit sharing, and immersion-aware provisioning within a cooperative multi-MSP model. The resource allocation problem is formulated as two NP-hard challenges: a non-cooperative setting where MSPs operate independently and a cooperative setting utilizing a General Credit Pool (GCP) for dynamic resource sharing. Using Deep Reinforcement Learning (DRL) for tuning resources and managing cooperating MSPs, CIVIC achieves 12-36\% higher request completion, 23-70\% higher fulfillment rates, 20-60\% more served clients, and up to 51\% more fairly distributed requests, all with competitive costs. Extensive experiments demonstrate CIVIC’s resilience, adaptability, and robust performance under dynamic load conditions and unexpected demand surges, making it suitable for real-world distributed Metaverse infrastructures.


\end{abstract}


\section{INTRODUCTION}
\label{sec:intro}
\IEEEPARstart{T}he Metaverse represents a paradigm shift in human–computer interaction, emerging as a persistent, immersive digital universe where physical and virtual realities converge. This transformative technology enables users to participate in various virtual rooms (VRooms) of different virtual environments (VEs), ranging from educational settings to entertainment venues, and interact with digital twins (DTs) of real-world objects \cite{9880528}. As the global Metaverse market is projected to reach \$936 billion by 2030, growing at a compound annual growth rate (CAGR) of 46\% \cite{grandview2025metaverse}, its applications span diverse domains including remote work, social interaction, gaming, education, and industrial training.

At the core of the Metaverse experience lies the critical challenge of immersion, which refers to the degree to which users feel present and engaged within their VRooms. Achieving high immersion levels requires seamless integration of multiple technological components, such as high-fidelity VEs, accurate DTs with real-time synchronization, and responsive user interactions. 

These requirements place unprecedented demands on computational and network resources, as Metaverse applications must maintain frame rates exceeding 60 frames per second (FPS), deliver ultra-low latency below 20 milliseconds, and support bandwidth-intensive operations reaching 4.6 terabits per second for educational applications alone \cite{alves2021beyond}.
The resource allocation challenge in the Metaverse is further complicated by its inherently heterogeneous and dynamic nature \cite{aaa}. VRooms vary significantly in their experience and resource requirements. For instance, a virtual art gallery demands high visual fidelity for static displays, while a gaming arena requires rapid response times and complex physics simulations \cite{bbb}. Moreover, user populations shift dynamically as users navigate between different VRooms according to their interests, resulting in sudden surges that may affect the overall experience. These dynamics create a complex optimization problem that Metaverse service providers (MSPs) must solve to maintain immersion and ensure the quality of perceptual experience (QoPE) for all users. The challenge becomes even more pronounced in multi-MSP scenarios, where multiple MSPs operate independently with varying budgets and resource capacities. Without coordination, this leads to inefficient resource utilization, where some MSPs struggle with demand surges, while others have idle capacity. Additionally, the complexity of jointly optimizing computational resources and network resources for VE and DT, while maintaining immersion thresholds across heterogeneous VEs, creates an NP-hard optimization problem (Proven later) that traditional approaches cannot solve effectively.


Although many articles in the literature have focused on Metaverse resource management (see Section \ref{sec:rlw}), existing approaches have not provided a holistic view of the problem and often neglect certain aspects. Many rely on oversimplified immersion models that use basic metrics to approximate quality without capturing the nuanced interplay between VE rendering, DT fidelity, and their requirements. Resource requirement scaling is often assumed to be linear, overlooking various aspects, including real-world nonlinearities such as density-dependent effects and heterogeneous room-specific demands. Additionally, they typically assume a single provider operating in isolation.

In this paper, we address these critical resource allocation challenges by proposing a Cooperative Immersion Via Intelligent Credit-sharing (CIVIC). This novel framework optimizes resource allocation and coordination across multiple MSPs to enhance users' immersion. 
Our key contributions can be summarized as follows:
\begin{itemize}
    \item First, to better quantify user immersion with DTs' fidelity, we introduce a more comprehensive immersion model that combines structural similarity, video streaming smoothness, and DT fidelity (structural, behavioral, and temporal).
    \item Secondly, we present novel non-linear resource scaling functions that incorporate density effects, resource sharing efficiency, and heterogeneous VRoom requirements to identify the needed resources for each VRoom.
    \item Then, we formulate the resource allocation challenge as two distinct non-convex NP-hard optimization problems designed to increase the number of fulfilled requests for all Heads with varying demands for the two deployment options: non-cooperative, where MSPs independently optimize allocation of different resources for their clients (i.e., Heads), and cooperative, which additionally incorporates inter-MSP resource sharing via a novel shared credit pool. 
    \item We then propose two Deep Reinforcement Learning (DRL) solutions that learn optimal policies for both non-cooperative and cooperative settings, while adhering to the system dynamics and different constraints. 
    \item To validate our solutions, rigorous testing and experimentation were performed on both solutions and showed significant improvements in request completion, fulfillment, resource balance, overall system efficiency, and adaptability compared to baseline methods.
\end{itemize}

Following this section, the paper proceeds as follows: Section \ref{sec:rlw} reviews relevant literature in the field. Section \ref{sec:ds_civic} introduces our system model, followed by Section \ref{sec:probform}, which formalizes the optimization problems. Section \ref{sec:ourmodel_civic} presents our DRL approach to address these challenges. Performance evaluation is conducted in Section \ref{sec:expeval}, and Section \ref{sec:conc_civic} concludes the work with the discussion of key findings and future directions.

\begin{table*}[t]
\centering
\caption{Table of Notations}
\label{tab:notations}
\begin{tabularx}{\textwidth}{lX}
\toprule
\textbf{Notation} & \textbf{Description} \\
\midrule
$\mathcal{M},\mathcal{H},\mathcal{V},\mathcal{C},\mathcal{T} $ & Set of MSPs, Heads, VRooms, VRCs and Timesteps \\
$\mathcal{C}_h \subseteq \mathcal{C}$ & Subset of VRCs linked to Head $h$. \\
$\mathcal{H}_m \subseteq \mathcal{H}$ & Set of Heads served by MSP $m$. \\
$\mathcal{D}_h$ & Set of IoT devices and sensors for Head $h$. \\
$I_h^{(t)}$ & Immersion level of Head $h$ at time $t$. \\
$w_q, w_f, w_a, w_{\sigma}, w_{\eta}$ & Weights for QOPE, frame rate, accuracy, SSIM and VMAF. \\
$\hat{Q}_h^{(t)}, \hat{F}_h^{(t)}, \hat{A}_h^{(t)}$ & Normalized QOPE, frame rate, and accuracy for head $h$ at $t$. \\
$Q_h,F_h,A_h$ & QoPE, frame rate, and Accuracy of DTs within the Head $h$. \\
$\sigma_h$ and $\eta_h$ & SSIM and VMAF. \\
$\omega_h$ & Average rotation speed of clients in Head $h$. \\
$B_h^{(t)}$ and $f_h^{(t)}$ & Bitrate and frame rate for Head $h$.\\
$\varepsilon_h^{(t)}, \beta_h^{(t)}, \tau_h^{(t)}$ & Structural, behavioral, and temporal accuracy for DTs inside $h$ VRoom. \\
$\mathbb{M}_\varepsilon, \mathbb{N}_\varepsilon$ & Sets for structural accuracy calculation. \\
$\mathbb{M}_\beta, \mathbb{N}_\beta$ & Sets for behavioral accuracy calculation. \\
$f_d, \lambda$ & Update frequency of DT $d$ and Rate of accuracy reduction. \\
${R_{s}^{y}}^{(t)}$ & General form of total scaled resource of type $R$ for a VRoom component $y$. Examples: ${Comp}_{b}^{VE}$ and ${Net}_{b}^{DT}$ \\
$R_{b}^{y}$ & General form of base resource of type $R$ for a VRoom component $y$ \\
$f_{v_{min}},f_{v_{max}}$ & Minimum and maximum VRoom frame rate. \\
$\kappa$ & Resource-specific scaling exponent (e.g., $\kappa_{Comp}\ \text{and}\ \kappa_{Net}$). \\
$C_h, V_{h_{max}}$ & Number of clients of the head $h$ and the maximum user capacity of VRoom $v$. \\
$\iota,\Lambda, \vartheta$  & Resource sharing efficiency factor, user density impact, and user scaling parameter. \\
$K_{Comp}$,$K_{Net}$ & unit prices for computing and network resources.\\
$\varphi$ & General Scaling Factor.\\
$P_v, O_v, A_v$ & VRoom's VE parameters: Polygon count, physical objects, interaction points. \\
$N_d, SV_d, U_d$ &VRoom's DT parameters: Number of sensors, state variables, update frequency. \\
$e_0, \dots, e_4$ & SSIM calculation coefficients \cite{vmafssim}. \\
$j_1, \dots, j_4$ & VMAF calculation coefficients \cite{vmafssim}. \\
${\delta_m^{(t)}}$ & {Donation fraction of MSP $m$'s post-allocation surplus.} \\
${z_h^{(t)}}$ & {Binary Head-level satisfaction indicator at time $t$.} \\
${S_m^{(t)}, W_m^{(t)}}$ & {Post-allocation local surplus and deficit of MSP $m$ at time $t$.} \\
${D_m^{(t)}}$ & {Actual donation of MSP $m$ to the GCP at time $t$.} \\
{$\mathcal{B}_{v_h}$} & {Discrete behavioral-accuracy grid associated with Head $h$'s VRoom.} \\
{$\Delta_\beta$} & {Step size of the behavioral-accuracy grid, set to $0.05$.} \\
$GCP^{(t)}$ & General Credit Pool's resources at time $t$. \\
${Budget_m^{(t)}}$ & {Remaining budget available to MSP $m$ at the beginning of time $t$.} \\
$L_m^{(t)}$ & Binary MSP-level request fulfillment indicator. \\
$\mathcal{H}_{m,\text{active}}^{(t)}$ & Set of active heads for MSP $m$ at time $t$. \\
$\mathbb{I}_{\text{request}}(h,t)$ & Request indicator function. \\
$I_{\text{threshold}}$ & Immersion threshold requirement. \\
$\phi(.)$ & Piece-wise efficiency function. \\
$w_{\text{imm}}, w_{\text{fin}}$ & Reward weights. \\
$\boldsymbol{P}$ and $\boldsymbol{\bar{P}}$& our optimization problem for non-cooperative and cooperative settings. \\
$\mathbb{S}, \mathbb{A}, \mathbb{P}, \ \text{and} \ \mathbb{R} $ & State space, Action space, transition probability, and reward. \\
\updated{MINLP} & {Mixed-Integer Nonlinear Programming.} \\
\updated{$r_{arr} ,r_{dep}$} & {Arrival and Departure in the double Poisson process.}\\
\bottomrule
\end{tabularx}
\end{table*}

\section{RELATED WORK}
\label{sec:rlw}
\begin{table*}[t]
\begingroup
    \centering
    \caption{\updated{Representative Related Work vs. our paper (Focus, Technique, and Key Capabilities)}}
    \label{tab:rw_compare}
    \renewcommand{\arraystretch}{1.08}
    \scriptsize
    \begin{tabularx}{\textwidth}{p{1.6cm} X X c c c c c c}
    \toprule
    \textbf{Ref.} & \textbf{\shortstack{Goal}} & \textbf{Technique} &
    \textbf{\shortstack{Comp\\+Comm}} & \textbf{\shortstack{Cross\\-MSP}} & \textbf{\shortstack{Incent.\\ /Credit}} &
    \textbf{\shortstack{Imm\\(VE+DT)}} & \textbf{\shortstack{Heterog.\\VRooms}} & \textbf{\shortstack{Nonlinear\\Scaling}} \\
    \midrule

    \cite{refeig} & Maximize Meta-Immersion (QoE) under xURLLC + contract constraints & Contract theory + optimization
    & \checkmark & $\times$ & \checkmark & \checkmark\ (VE) & $\times$ & $\times$ \\
    \cite{refelven} & Maximize QoE via attention-aware network resource allocation & Analytical model + optimization
    & $\times$ & $\times$ & $\times$ & \checkmark\ (VE) & $\times$ & $\times$ \\
%
    \cite{reftweelve} & Maximize long-term VR QoE under interaction-latency constraint & DRL
    & \checkmark & $\times$ & $\times$ & \checkmark\ (VE) & $\times$ & $\times$ \\
    \cite{feng2024resource} & Maximize MI while minimizing device energy (Pareto trade-off) & Multi-agent multi-objective RL
    & \checkmark & $\times$ & $\times$ & \checkmark\ (VE) & $\times$ & $\times$ \\
    \cite{9880566} & Maximize MSP utility via reliable CDC rendering (mitigate stragglers) & Blockchain + game-theoretic design
    & \checkmark & $\times$ & \checkmark & $\times$ & $\times$ & $\times$ \\
    \cite{10250875} & Maximize data utility while minimizing vehicular energy (AR Metaverse) & Optimization
    & \checkmark & $\times$ & $\times$ & $\times$ & $\times$ & $\times$ \\
    \cite{chu2024dynamic,aaa} & Maximize long-term provider revenue/acceptance via multi-tier allocation & DRL (sMDP)
    & \checkmark & $\times$ & $\times$ & $\times$ & \checkmark & \checkmark \\
%
    \cite{reftwenyfour} & Minimize delay and maximize video quality to meet target QoE & Multi-objective DRL
    & \checkmark & $\times$ & $\times$ & \checkmark\ (VE) & $\times$ & $\times$ \\
    \cite{10294071} & Maximize VR streaming QoE under bandwidth via robust FoV tile selection & Contextual bandit
    & $\times$ & $\times$ & $\times$ & \checkmark\ (VE) & $\times$ & $\times$ \\
%
    \cite{9838736} & Maximize social welfare while minimizing auction overhead (VR services) & Double Dutch auction + DRL
    & \checkmark & $\times$ & \checkmark & \checkmark\ (VE) & $\times$ & $\times$ \\
%
    \textbf{Our paper} & \textbf{Maximize holistic immersion (VE+DT) under multi-MSP budgets via cooperation} &
    \textbf{DRL + shared-credit cooperation}
    & \textbf{\checkmark} & \textbf{\checkmark} & \textbf{\checkmark} & \textbf{\checkmark\ (VE+DT)}
    & \textbf{\checkmark} & \textbf{\checkmark} \\
    \bottomrule
    \end{tabularx}
    \par\addvspace{5pt}
    \noindent{\scriptsize Comp+Comm: jointly considers computing and communication. Cross-MSP: cooperation across distinct MSP administrative domains. Heterog. VRooms: explicitly models room types with different base requirements. Nonlinear Scaling: models resource growth with user density/room complexity beyond linear forms.}
\endgroup
\end{table*}

\updated{The Metaverse differs from conventional applications (e.g., online gaming and video playback) because it requires persistent, bidirectional interaction between users and a shared virtual world under stringent latency, bandwidth, and compute constraints \cite{9880528}. Consequently, closely related work can be organized around two tightly coupled questions: \emph{(i) how immersion/QoE is quantified} (i.e., what the system should optimize), and \emph{(ii) how resources are allocated to realize that immersion under constraints} (i.e., how the system should provision compute and networking). Table~\ref{tab:rw_compare} summarizes representative efforts and highlights the gap addressed by CIVIC.

On the \emph{immersion calculation} side, prior studies have proposed QoE-style indicators spanning network KPIs (rate, loss, and latency) \cite{reffif,reftwennine,refthirty,refeig}, media/device factors (rendering delay, resolution, and video-quality measures) \cite{refelven,reftweelve,reftwenyfour}, and perception/engagement variables tied to content and encoding decisions \cite{shamim2025evaluating,anwar2020impact,desai2017qoe,liu2023qoe}. These formulations motivate immersion-centric optimization, but they are commonly \emph{VE-centric}: they primarily capture what users visually perceive and how the VE is delivered. In parallel, DT work has developed fidelity/identicality dimensions and evaluation methodologies \cite{dts,dts2,nostand,nostand2}, yet DT fidelity is often treated as a separate modeling/evaluation concern rather than an explicit term in immersion-driven provisioning. CIVIC bridges this split by making immersion explicitly \emph{VE--DT coupled}, integrating VE perceptual quality with DT structural/behavioral/temporal fidelity within a unified objective (Table~\ref{tab:rw_compare}, columns ``Imm. Metric'' and ``DT Fidelity'').

On the \emph{resource allocation} side, a large body of work optimizes communication and/or computing resources to improve QoE/immersion. Network-centric schemes study attention-aware/customized provisioning and radio-side QoE optimization \cite{9999298,refelven,refeig}. Rendering/offloading research explores where computation should occur, including MEC-enabled VR optimization \cite{reftweelve}, multi-objective learning that trades immersion against device energy \cite{feng2024resource}, collaborative/reliable computation for rendering \cite{9880566}, and vehicular AR Metaverse resource optimization under energy constraints \cite{10250875}. Streaming-oriented work improves VE delivery via FoV/tile decisions and constrained learning \cite{10294071,reftwenyfour}. Broader multi-tier frameworks incorporate admission/placement and cross-tier allocation using DRL \cite{chu2024dynamic,aaa}. While these directions substantially advance immersion-driven provisioning, they typically assume a single provider/controller domain, treat DT synchronization quality as outside the optimized experience metric, and do not explicitly model heterogeneous VRooms with distinct VE/DT primitives and \emph{nonlinear} scaling with user density and fidelity targets. CIVIC targets this missing intersection by (i) jointly provisioning VE and DT resources to maximize holistic immersion, (ii) explicitly modeling heterogeneous VRooms and nonlinear scaling laws that govern both rendering and DT synchronization cost, and (iii) addressing a multi-MSP setting with time-varying localized scarcity via a shared-credit cooperation mechanism that enables cross-MSP resilience (Table~\ref{tab:rw_compare}, columns ``Cross-MSP'' and ``Incent./Credit'').}

\begin{figure*}
    \centering
    \includegraphics[width=0.7\linewidth]{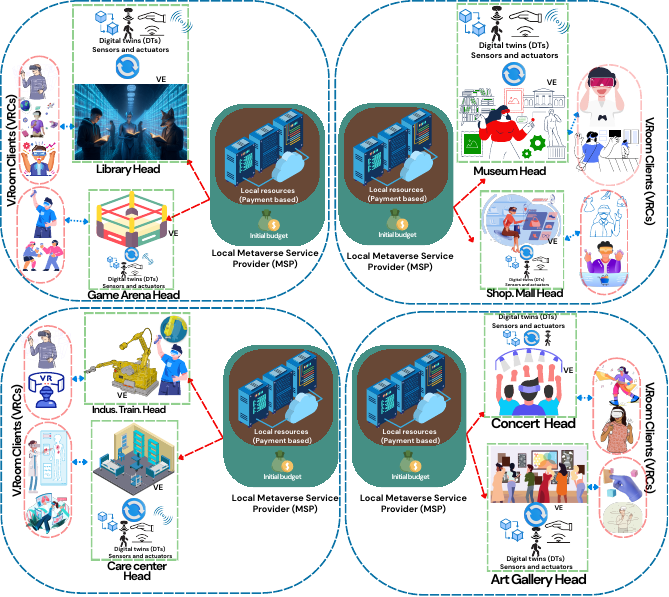}
    \caption{The main system model, composed of geographically distributed Metaverse service providers, serving multiple Heads who host virtual Rooms of their physical entity and enhance their clients' immersion with their DTs synchronization}
    \label{fig:sysmod}
\end{figure*}

\section{SYSTEM MODEL}
\label{sec:ds_civic}
In this section, we introduce the proposed system model as in \fref{fig:sysmod}. We begin by describing the main entities, their terminologies, and goals, followed by the two problem formulations, $\boldsymbol{P}$ and $\boldsymbol{\bar{P}}$, which integrate these components together for both non-cooperative and cooperative settings, respectively. For better readability, the table of notations is supplied in \tref{tab:notations}.

\subsection{Network Architecture and Entities}

As shown in \fref{fig:sysmod}, our Metaverse environment comprises four primary entities: Metaverse Service Providers (MSPs), Cluster Heads (Heads), Virtual Rooms (VRooms), and Virtual Room Clients (VRCs).

\paragraph{Metaverse Service Providers (MSPs)} 
MSPs serve as essential gateways, enabling user access to the Metaverse by delivering the virtualized computing and networking infrastructure required for Heads to render and operate their VRooms and DTs. Each MSP operates with an initial budget (allocated by higher management) that provides access to payment-based cloud resources, securing the computational capacity and network connectivity needed to support Metaverse interactions. The MSPs are strategically positioned across different geographical locations and can operate either as independent entities or as collaborative units under the management of an organization. In our case, we considered both scenarios in $\boldsymbol{P}$ and $\boldsymbol{\bar{P}}$. We represent the complete set of MSPs mathematically as $\mathcal{M} = \{1, \dots, m, \dots, M\}$, where each element $m$ denotes a specific MSP and $M$ indicates the total number of MSPs deployed in the system. In our system, we assume that each MSP serves multiple Heads.

\paragraph{Cluster Heads (Heads) and Virtual Rooms (VRooms)} Heads are physical establishments such as libraries, gaming arenas, or art galleries, which virtualize their spaces in the Metaverse for various purposes, including monetary gain from virtual services and enhanced user experience. Each Head converts its real-world physical environment into a digital VRoom through different methods, such as 3D scanning and reconstruction. Similar to MSPs, the set of all Heads is denoted as $\mathcal{H} = \{1, \dots, h, \dots, H\}$, where $h$ refers to a specific Head and $H$ is their count within the system. 

\updated{
VRooms, digital counterparts of the Heads' physical spaces, offer virtual interaction opportunities to Virtual Room Clients (VRCs). Each VRoom consists of two coupled subsystems: (i) the \emph{Virtual Environment (VE)}, i.e., the rendered 3D scene and interactive logic that VRCs visually perceive and navigate (layout, objects, avatars, and viewpoint updates), and (ii) the \emph{Digital Twin (DT)} layer, i.e., the synchronized digital representation of physical objects inside the Head’s venue enabled by IoT devices, sensors, and actuators (e.g., cameras/LiDAR/touch sensors and haptics) to provide real-time state alignment and interaction feedback. The set of VRooms is defined as $\mathcal{V} = \{v_1, \dots, v_h, \dots, v_H\}$ explicitly mapping each Head $h$ to its unique VRoom $v_h$. We envision that each type of VRoom (e.g., library and game arena) has unique resource requirements that impact service quality, scalability, and cost, among others. The key VRoom parameters we proposed are shown in \tref{tab:vrooms_params}. The set of DTs available for the Head $h$ is denoted as $\mathcal{D}_h = \{1, \dots, d, \dots, D_h\}$, where $d$ represents an arbitrary DT within Head $h$’s venue and $D_h$ is the total number of DTs in $h$’s VRoom. Since Heads typically lack adequate computational power to render their VRooms and maintain DT synchronization at immersive quality, they depend on MSPs for these services. Each Head is exclusively served by an MSP. We define $\mathcal{H}_m \subseteq \mathcal{H}$ as the set of Heads served by MSP $m$, satisfying $\sum_{m=1}^{M} |\mathcal{H}_m| = H$.
}

\paragraph{Virtual Reality Clients (VRCs)} VRCs are users who access and interact with the Heads' VRooms in the Metaverse, primarily using HMDs and VR sets. The complete set of VRCs is indicated as $\mathcal{C} = \{1, \dots, c, \dots, C\}$, with the subset of VRCs linked to the Head $h$ symbolized by $\mathcal{C}_h \subseteq \mathcal{C}$. Each VRC dynamically navigates from one VRoom to another based on their interest.

\updated{To summarize, our environment comprises multiple MSPs that serve Heads, who in turn attract various VRCs to their respective VRooms. This hierarchy reflects practical Metaverse deployment and service provisioning, while enabling tractable control in a multi-provider setting. MSPs allocate virtualized compute and network resources at the Head level, where a common virtual view is rendered for all associated VRCs and where budgets and DT-enabled VRooms are managed, rather than at the individual user level, which aligns with realistic B2B service models. Assuming one VRoom type per Head ensures a well-defined resource–immersion profile and isolates heterogeneity across environments; venues supporting multiple room types can be equivalently represented by multiple logical Heads or time-varying VRoom profiles.}

In the following subsections, we present critical resources and quality metrics along with associated costs.

\subsection{Immersion, Resource Scaling and Cost Models}
\subsubsection{Immersion}
Immersion refers to the degree of presence and engagement experienced by a VRC within the VRoom. Although no standardized metric or list of metrics exists for calculating the immersion of a user inside the Metaverse \cite{nostand,nostand2}, we envision that an immersive experience involves two primary components: VE quality (i.e., visual quality and motion smoothness of the VRCs inside the VE) and DT fidelity (i.e., DT look, feel and feedback, synchronization including sound, vibration among others).

When a Head $h$ requests the rendering services for VRoom $v_h$ from MSP $m$ at time $t$, the immersion experienced by $h$, denoted as $I_h^{(t)}$ is calculated as a weighted sum (the weights, $w_{q}+ w_{f}+w_{a} = 1$, adjusted according component importance ) of three normalized metrics as:

\begin{equation}
I_h^{(t)} = w_{q}\hat{Q}_h^{(t)} + w_{f}\hat{F}_h^{(t)} + w_{a}\hat{A}_h^{(t)}
\end{equation}

Here, $\hat{Q}_h$ represents the normalized quality of perceptual experience (QoPE) based on VRoom-specific visual and streaming requirements and is defined as:

\begin{equation}
Q_h = w_{\sigma}\sigma_h^{(t)} + w_{\eta}\eta_h^{(t)}
\end{equation}

$\sigma_h$ denotes structural similarity index (SSIM) and evaluates image quality (e.g., structure, texture, clarity), while $\eta_h$ video multi-method assessment fusion (VMAF) assesses video streaming smoothness and clarity, both are defined in \cite{vmafssim} as:

\begin{equation}
\label{eq:ssim}
\begin{aligned}
\sigma_h^{(t)} = &\max\left(e_0, 1 - (e_1 + e_2 \omega_h) B_h^{(t)^{-(e_3 + e_4 \omega_h)}} \right)\\
    \eta_h^{(t)} = & \min\left(100, j_1 + j_2 \omega_h + j_3 B_h^{(t)} + j_4 \omega_h B_h^{(t)} \right) 
\end{aligned}
\end{equation}
Herein, $e_0, \dots, e_4$ and $j_1, \dots, j_4$ are the coefficients for SSIM and VMAF, $\omega_h$ represents the average rotation speed of clients in the Head's VRoom and is considered a fixed parameter, $B_h^{(t)}$ is the bitrate (in Mbps) allocated for Head $h$ (by MSP $m$) at time $t$. $w_{\sigma}$ and $w_{\eta}$ are prioritization weights that sum to 1 (i.e., $w_{\sigma}+w_{\eta} = 1$), defined per VRoom type (e.g., gaming arenas prioritize VMAF, while art galleries prioritize SSIM).

$\hat{F}_h$ represents the normalized responsiveness of the VRoom, calculated from the actual frame rate $f_h$, normalized against the VRoom's specified minimum and maximum frame rates.

Finally, $A_h^{(t)}$ captures the (unnormalized) accuracy of DTs within Head $h$ at time $t$, computed as the harmonic mean of structural accuracy $\varepsilon_h^{(t)}$, behavioral accuracy $\beta_h^{(t)}$, and temporal accuracy $\tau_h^{(t)}$:

\begin{equation}
\label{eqn:dtacc}
A_h^{(t)} = \begin{cases}
0 & \text{if } \varepsilon_h^{(t)} = 0 \ \text{or}\ \beta_h^{(t)} = 0 \ \text{or}\ \tau_h^{(t)} = 0 \\
\frac{3}{\frac{1}{\varepsilon_h^{(t)}} + \frac{1}{\beta_h^{(t)}} + \frac{1}{\tau_h^{(t)}}} & \text{otherwise}
\end{cases}
\end{equation}

\updated{Structural accuracy \( \varepsilon_h^{(t)} \) measures how faithfully the DT captures the physical system's required structure at time \( t \), defined as the fraction of correctly represented structural features:
\[
\varepsilon_h^{(t)} = \frac{|\mathbb{M}_\varepsilon^{(t)}|}{|\mathbb{N}_\varepsilon|}.
\]
Here, \( \mathbb{N}_\varepsilon \) is the set of structural DT features required by the application (e.g., geometry primitives, key mesh points, textures, calibrated dimensions, and pose anchors), and \( \mathbb{M}_\varepsilon^{(t)} \subseteq \mathbb{N}_\varepsilon \) is the subset whose DT representation matches the physical reference within an application-defined tolerance (e.g., geometric error \( \leq \epsilon_{\text{geo}} \) and texture similarity \( \geq \epsilon_{\text{tex}} \)).

Behavioral accuracy \( \beta_h^{(t)} \) measures how faithfully the DT reproduces required interaction/behavioral responses, defined as:
\[
\beta_h^{(t)} = \frac{|\mathbb{M}_\beta^{(t)}|}{|\mathbb{N}_\beta|},
\]
where \( \mathbb{N}_\beta \) is the set of behavioral interaction features (e.g., state-transition rules, actuator response modes, contact/haptic response patterns, and dynamic constraints), and \( \mathbb{M}_\beta^{(t)} \subseteq \mathbb{N}_\beta \) is the subset whose simulated response matches the physical system within tolerance (e.g., response delay \( \leq \epsilon_{\text{delay}} \), force/torque error \( \leq \epsilon_{\text{force}} \)).

Temporal accuracy \( \tau_h^{(t)} \) captures DT freshness via an increasing, saturating function of DT update frequency:
\[
\tau_h^{(t)} = \frac{1}{D_h} \sum_{d=1}^{D_h} \left(1 - e^{-\lambda f_d^{(t)}}\right),
\]
where \( f_d^{(t)} \) is the update frequency of DT \( d \), and \( \lambda \) characterizes the accuracy reduction rate (uniform across DTs in a VRoom for simplicity). The sets \( \mathbb{N}_\varepsilon \) and \( \mathbb{N}_\beta \) are specified at VRoom instantiation and can be extended to incorporate domain-specific fidelity metrics.

We use the harmonic mean in \eqref{eqn:dtacc} to reflect that DT fidelity is a \emph{weakest-link} property: if any of \( \varepsilon_h^{(t)}, \beta_h^{(t)}, \tau_h^{(t)} \) is low, overall DT fidelity degrades sharply even when the others are high. The zero case captures a non-functional DT state (no structural/behavioral/temporal alignment), yielding zero DT fidelity contribution to immersion. Finally, in our model \( \beta_h^{(t)} \) is a controllable DT service level: increasing it raises DT compute/network cost (shown later) but improves immersion through \( \hat{A}_h^{(t)} \).}
%



In $\hat{Q}_h \ , \hat{F}_h \ \text{and} \ \hat{A}_h$ normalization is performed via min-max scaling (i.e., $\hat{x} = \frac{x - x_{\min}}{x_{\max} - x_{\min}}$) on the VRooms' minimum and maximum parameters as found in \tref{tab:vrooms_params}.

\updated{Finally, it is important to note that we do not treat VE and DT as independent modules. They are \emph{coupled} through both perception and resources: (a) \emph{perceptual coupling}, because user immersion is computed jointly from VE visual quality/motion and DT fidelity, so degradation in either subsystem reduces the same immersion score; and (b) \emph{resource coupling} (as shown later), because improving VE quality (e.g., higher bitrate $B_h^{(t)}$ and frame rate $f_h^{(t)}$) and improving DT fidelity (e.g., higher behavioral accuracy $\beta_h^{(t)}$ and DT synchronization traffic) both increase compute/network requirements and compete under the same MSP budget and nonlinear scaling constraints. Consequently, allocating more resources to VE rendering can reduce feasible DT fidelity and vice versa, which motivates our joint VE--DT immersion-aware optimization.}

\updated{Finally, we note that the immersion weights ($w_{q},w_{f},w_{a}, w_{\sigma} \  \text{and} \ w_{\eta}$) are VRoom-type parameters that encode application-level priorities rather than optimization variables. They are fixed during training (with values as in \tref{tab:vrooms_params}) and reflect domain semantics (e.g., motion sensitivity in arenas versus visual fidelity in galleries). This modular formulation preserves interpretability and allows heterogeneous immersion profiles without altering the underlying optimization or learning framework.}

\subsubsection{Resource Scalability}
Since each Head's VRoom $v_h$ consists of the VE that needs to be rendered and a set of DTs that needs to be rendered and synchronized, it is then of paramount importance to know how these resources's computational and network requirements will be scaled as more VRCs join the VRoom and/or the visual quality is demanded (i.e., more frame rate is required).

Inspired by recent advancements in distributed systems and resource optimization \cite{czarnul2025optimization,bhaskaran2025comprehensive}, our resource scaling model incorporates non-linear user scaling \cite{czarnul2025optimization}, resource sharing efficiency (i.e., degree of sharing a resource, value between 0-1)\cite{awad2025resource},  user density effects (i.e., how full VRoom affects the resource scaling) \cite{kiggundu2024resource,delgado2024optimal}, aligning with current research in cloud computing and edge environments \cite{KareemAwadZainolAriffinNazriYassen+2025,liang2024resource}.
The proposed resource scaling model for any resource for any VRoom component is presented as:

\begin{align}
\label{eqn:scaling}
{R_{s}^{y}}^{(t)} &= 
\underbrace{\left[{R_{b}^{y}}^{(t)} \cdot \left(\frac{f_h^{(t)}}{f_{v_{\min}}}\right)^{\kappa} \right]}_{\text{Base resource needed}}  \cdot  \notag \\
&\quad
\underbrace{\left[C_h^{(t)\vartheta} \cdot \left(1 - \iota \left( 1 - \left( \frac{C_h^{(t)}}{V_{h_{\max}}} \right)^{\Lambda} \right) \right) \right]}_{\text{Scale of res. accord. to res. and VRoom nature}}
\end{align}
 
Where ${R_{s}^{y}}^{(t)}$ represents the scaled resource of type $R$ at time $t$, (Herein, $R \in \{Comp, Net\}$ represents computational and network resources respectively) for a specific VRoom's component $y$, ($y \in \{VE, DT\}$ represent VE and DT) at a specific type of VRoom. From \eqref{eqn:scaling}, ${R_{s}^{y}}^{(t)}$ is composed of two main terms:

The first term, presented by $[{R_{b}^{y}}^{(t)} \cdot (\frac{f_h^{(t)}}{f_{v_{min}}})^\kappa]$ signifies the base resources of type $R$ required for a VRoom component $y$ with current frame rate relative to the minimum room's frame rate.
Applying ${R_{b}^{y}}^{(t)}$  in our resources and VRoom components, we get 1) base Computational resource requirements for a VE and DT denoted as ${Comp}_{b}^{VE}$ and ${Comp}_{b}^{DT}$ and 2)  base network resources for VE and DT denoted as ${Net}_b^{VE}$ and ${Net}_b^{DT}$.
     
${{Comp}_{b}^{VE}}^{(t)}$ is calculated as: \[{{Comp}_{b}^{VE}}^{(t)} = \frac{P_v^{(t)}}{P_{max}} + \frac{{O_v^{(t)}}}{{O}_{max}} + \frac{{A_v^{(t)}}}{{A}_{max}}\] where $P_v^{(t)}, O_v^{(t)}$ and $A_v^{(t)}$ are the parameters affecting the base resources of VE of VRoom $v$ at $t$, which are polygon count $P_v$, number of physical objects $O_v$ and interaction points, $A_v$ while $P_{max}, O_{max} \ \text{and} \ A_{max} $ are their maximum values fixed across all VRooms $\mathcal{V}$ and available in \tref{tab:vrooms_params}. 
Additionally, the base network resources for VE can be described as the minimum per-VRoom fixed bitrate ${B_v}^{(t)}_{min}$: \[
{{Net}_{b}^{VE}}^{(t)} = {B_v}^{(t)}_{min}
    \]
    
Similarly, building upon DTs' requirements discussed in \cite{dts,dts2}, ${{Comp}_{b}^{DT}}^{(t)}$ can be described as: \[{{Comp}_{b}^{DT}}^{(t)} = \frac{N_d^{(t)}}{N_{max}} + \frac{SV_d^{(t)}}{SV_{max}} + \frac{U_d^{(t)}}{U_{max}}\] where $N_d^{(t)}, SV_d^{(t)} \ \text{and} \ U_d^{(t)}$ represent the factors affecting the complexity of DT $d$, and are the number of sensors, the number of DT's state variables, and DT's synchronization update frequency.  $N_{max}, SV_{max} \ \text{and} \ U_{max}$ are their maximum values, unified across all of the DTs $\mathcal{D}$. 
Moreover, the base network resources for DT can be described as a scaled-down version of VE's network resources multiplied by a scaling factor $\varphi$ and the DT's behavioral accuracy $\beta$, therefore \[
{{Net}_{b}^{DT}}^{(t)} = (\varphi \cdot {B_v}^{(t)}_{min}) \cdot \beta_h^{(t)} 
\]
Further, since all of the aforementioned base resources are affected by the smoothness of the visual quality, we multiply them by a visual quality scaling factor $(\frac{f_h^{(t)}}{f_{v_{min}}})^\kappa$ representing the ratio of head's frame rate $f_h^{(t)}$ ( one of the decision variables, as explained later) divided by the VRoom's minimum frame rate $f_{v_{min}}$, and $\kappa \in \{Comp, Net\}$ is a resource-specific exponent that shows how this resource scales with visual smoothness. 

 The second term of ${R_{s}^{y}}^{(t)}$ calculation $[{C_h}^{(t)\vartheta} \cdot (1 - \iota ( 1 - ( \frac{C_h^{(t)}}{V_{h_{max}}} )^{\Lambda} ) ) ]$ describes the scaling of resources according to 1) the resource nature (e.g., resource scaling and sharing capability) and 2) the VRoom's nature (e.g., VRC count effect on VRoom's capacity). Herein, $\vartheta$ controls the non-linear scaling behavior as VRC count ($C_h$) increases, similar to scaling effects observed in distributed systems, $\iota$ represents the resource-specific sharing efficiency factor, capturing how effectively resources can be shared as users increase, comparable to resource optimization approaches in cloud-based rendering systems and $\Lambda$ models the impact of user density relative to VRoom's maximum capacity (${V_{h_{max}}}$), reflecting congestion effects.

\subsubsection{Cost}
In the previous section, we explained how the different network and computational resources scale for the VE and DT. In this section, we will utilize the scaled resources across the different Heads belonging to each MSP to calculate the overall resource allocation cost. 

The total computational cost of an MSP $m$ across all of its heads can be computed as:

\[ \smallmath[0.85]{{{Cost}_{\text{Comp}}^{m}}^{(t)} = \sum_{h \in \mathcal{H}_m} K_{Comp} \cdot \left( Comp_s^{VE}(m, h)^{(t)} + Comp_s^{DT}(m, h)^{(t)} \right)} \]
Moreover, the total network cost is computed as:
\[\smallmath[0.9]{{{Cost}_{\text{Net}}^{m}}^{(t)} =\sum_{h \in \mathcal{H}_m} K_{Net} \cdot \left( {Net_s^{VE}(m, h)}^{(t)} + {Net_s^{DT}(m, h)}^{(t)} \right)}\]
where $K_{Comp}$ and $K_{Net}$ are the unit prices for computational and network resources, respectively.

Therefore, the total cost for an MSP $m$ evaluated at any arbitrary timestep $t$, denoted as ${Cost_{Total}^{m}}^{(t)}$ can be presented as the summation of computational and network resources as:
 \begin{equation}
 {Cost_{Total}^{m}}^{(t)} = {{Cost}_{\text{Comp}}^{m}}^{(t)}  + {{Cost}_{\text{Net}}^{m}}^{(t)} 
 \end{equation}  
After introducing the system model, the next section will discuss the problem formulation for both the non-cooperative and cooperative scenarios.

\begin{table}[ht]
    \centering \updated{
    \caption{Parameters for different VRooms}
    \label{tab:vrooms_params}
    \begin{tabularx}{\linewidth}{lccc}
        \toprule
        \textbf{Parameter} & \textbf{LIBRARY} & \textbf{ARENA} & \textbf{GALLERY} \\
        \midrule
        $B_{v_{min}}$, $B_{v_{max}}$ & 20, 25 & 30, 50 & 25, 35 \\
        $f_{v_{min}}$, $f_{v_{max}}$ & 30, 60 & 60, 120 & 30, 60 \\
        $\varepsilon_{v_{min}}$, $\varepsilon_{v_{max}}$ & 0.6, 1.0 & 0.5, 1.0 & 0.8, 1.0 \\
        $\beta_{v_{min}}$, $\beta_{v_{max}}$ & 0.5, 1.0 & 0.5, 1.0 & 0.3, 1.0 \\
        $\omega_h$, $\kappa_{Comp}$, $\kappa_{Net}$ & 400, 1.1, 0.5 & 720, 1.1, 0.5 & 400, 1.1, 0.5 \\
        $w_{\sigma}$, $w_{\eta}$ & 0.5, 0.5 & 0.3, 0.7 & 0.7, 0.3 \\
        $\vartheta_{Comp}$, $\vartheta_{Net}$ & 0.7, 0.9 & 0.85, 0.85 & 0.7, 0.8 \\
        $\Lambda$, $\iota$ & 0.8, 0.6 & 0.8, 0.5 & 0.7, 0.8 \\
        $P_v$, $O_v$, $A_v$ & $2 e5$, 50, 5 & $5 e5$, 200, 30 & $3 e5$, 20, 15 \\
        $N_d$, $SV_d$, $U_d$ & 50, 30, 2 & 200, 100, 10 & 30, 20, 1 \\
        $v_{h_{max}}$ & 10 & 100 & 10 \\
        $w_{q},w_{f},w_{a}$ & 0.33 & 0.33  & 0.33 \\
        $r_{arr},r_{dep}$ & 0.4, 0.7 & 0.4, 0.7 & 0.4, 0.7\\
        \bottomrule
    \end{tabularx}
    }
\end{table}

\begin{table}[h]
\centering
\updated{
\caption{General parameter values}
\begin{tabularx}{\linewidth}{lX}
\toprule
\textbf{Parameter} & \textbf{Value} \\
\midrule
$N_{\text{max}}, SV_{\text{max}}, U_{\text{max}}$ & 100, 100, 10 \\ 
$e_0, \ldots, e_4$ \cite{vmafssim} & 0.65, 0.368, $1.23 \times 10^{-3}$, 0.85, $1.23 \times 10^{-3}$ \\ 
$j_0, \ldots, j_4$ \cite{vmafssim} & 36.13, $-1.66 \times 10^{-2}$, 11.62, $-6.07 \times 10^{-3}$ \\ 
$K_{\text{Comp}}, K_{\text{Net}}$ & 0.01, 0.01 \\ 
$P_{\text{max}}, O_{\text{max}}, A_{\text{max}}$ & $1 \times 10^5$, 100, 10 \\ 
$\varphi, \tau_h^{(t)}$ & 0.1, 1 \\ 
$\mathcal{V}$ used sequentially & Lib., Arena ,Gal., Lib, Gal., Arena, Gal., Lib., Lib.\\
PPO Act/cri. NN & 2 layers of 128 neurons\\
PPO gamma & 0.97\\
PPO batch size & 128\\
PPO learning rate & $3\times 10^{-4}$\\
\end{tabularx}
\label{tab:generalparams}
}
\end{table}

    

\section{PROBLEM FORMULATION}
\label{sec:probform}
{In this section, we formalize the two optimization problems underlying CIVIC. We begin with a non-cooperative benchmark, in which each MSP operates using only its own budget, and then extend it to a cooperative formulation coupled through a shared General Credit Pool (GCP). To keep the notation compact, for every Head $h$ and timestep $t$, we define the local control vector}
{
\[
\mathbf{u}_h^{(t)} \overset{\mathrm{def}}{=} \bigl(B_h^{(t)}, f_h^{(t)}, \beta_h^{(t)}\bigr),
\]
}
{where $B_h^{(t)}$ is the allocated bitrate, $f_h^{(t)}$ is the frame rate, and $\beta_h^{(t)}$ is the DT behavioral-accuracy level selected from a finite fidelity grid. For each Head $h$, let $\mathcal{B}_{v_h}$ denote the admissible behavioral-fidelity grid associated with its VRoom type; throughout this paper, we use a uniform spacing $\Delta_\beta=0.05$ between consecutive levels. We also let $Budget_m^{(t)}$ denote the remaining budget available to MSP $m$ at the \emph{beginning} of slot $t$.}
\subsection{Problem $\boldsymbol{P}$: Non-cooperative setting}

The non-cooperative problem serves as the baseline case in which MSPs act independently over the horizon $\mathcal{T}=\{1,\ldots,T\}$. At each slot, MSP $m$ observes the set of Heads requesting service and chooses $\mathbf{u}_h^{(t)}$ for every served Head $h\in\mathcal{H}_m$. Each Head request specifies its VRoom type, the number of VRCs, the number of DTs within that VRoom and the required immersion level. A Head-level request is considered satisfied when its attained immersion meets or exceeds the common threshold $I_{\text{threshold}}$.
Moreover, each MSP operates under a fixed budget, and serving requests consumes budget over time; hence, MSPs must tune their allocation to control spending while maximizing the number of satisfied requests across $\mathcal{T}$.
To avoid starvation among Heads attached to the same MSP, we adopt an \emph{all-or-nothing} MSP-level service rule: a request at time $t$ is considered fulfilled for MSP $m$ only when \emph{all} of its active Heads are simultaneously are served at or above threshold. The initial condition is $Budget_m^{(1)}=Budget_m^{(0)}$ for every MSP $m$.

{To translate the above service rule into a compact mathematical form, we first introduce a Head-level satisfaction indicator and then use it to express the MSP-level objective. Specifically, let}
{
\[
z_h^{(t)} \overset{\mathrm{def}}{=} \mathbb{I}_{\text{request}}(h,t)\,\mathbf{1}\!\left[I_h^{(t)}\!\left(\mathbf{u}_h^{(t)}\right)\ge I_{\text{threshold}}\right].
\]
}
{Therefore, the non-cooperative optimization problem can be written as follows:}
{
\begin{align}
\boldsymbol{P}: \quad \max_{\{\mathbf{u}_h^{(t)}\}} \quad & \sum_{t=1}^{T} \sum_{m=1}^{M} L_m^{(t)} \label{eq:objective}
\end{align}
}

{\textbf{Subject to:}}
{
\begin{align}
& \mathcal{H}_{m,\text{active}}^{(t)} = \{h \in \mathcal{H}_m : \mathbb{I}_{\text{request}}(h,t) = 1\},
\quad \forall m \in \mathcal{M}, \forall t \in \mathcal{T} \tag{c1} \label{eq:active_heads} \\
& \mathbb{I}_{\text{request}}(h,t) \in \{0,1\},
\quad \forall h \in \mathcal{H}, \forall t \in \mathcal{T} \tag{c2}\label{eq:request_indicator} \\
& L_m^{(t)} = \mathbf{1}\!\left[|\mathcal{H}_{m,\text{active}}^{(t)}|>0\right]
\mathbf{1}\!\left[\sum_{h \in \mathcal{H}_{m,\text{active}}^{(t)}} z_h^{(t)} = |\mathcal{H}_{m,\text{active}}^{(t)}|\right], \nonumber\\
& \hspace{6.4cm} \forall m \in \mathcal{M}, \forall t \in \mathcal{T} \tag{c3}\label{eq:L_m_definition} \\
& {Cost_{Total}^{m}}^{(t)} \leq {Budget_{m}^{(t)}},
\quad \forall m \in \mathcal{M}, \forall t \in \mathcal{T}\tag{c4}\label{eq:budget_constraint} \\
& B_h^{(t)} \in \mathbb{Z}, \ B_{v_{\min}} \leq B_h^{(t)} \leq B_{v_{\max}},
\quad \forall h \in \mathcal{H}, \forall t \in \mathcal{T}\tag{c5} \label{eq:bitrate_bounds}\\ 
& f_h^{(t)} \in \mathbb{Z}, \ f_{v_{\min}} \leq f_h^{(t)} \leq f_{v_{\max}},
\quad \forall h \in \mathcal{H}, \forall t \in \mathcal{T}\tag{c6} \label{eq:framerate_bounds}\\
& \beta_h^{(t)} \in \mathcal{B}_{v_h},
\quad \forall h \in \mathcal{H}, \forall t \in \mathcal{T} \tag{c7}\label{eq:accuracy_bounds}\\ 
& Budget_m^{(t+1)} = Budget_m^{(t)} - {Cost_{Total}^{m}}^{(t)}, \nonumber\\
& \hspace{5cm} \forall m \in \mathcal{M}, \forall t \in \mathcal{T}\setminus\{T\}\tag{c8}\label{eq:budget_dynamics_noncoop}
\end{align}
}

{The constraints can be read naturally from request activation to budget evolution. Constraint \eqref{eq:active_heads} defines the active request set of MSP $m$ at slot $t$, while \eqref{eq:request_indicator} enforces the binary nature of the request indicator itself. Constraint \eqref{eq:L_m_definition} then captures the key structural choice in CIVIC: the objective counts an MSP as successful at slot $t$ only if every active Head attached to it is simultaneously satisfied, which creates coupling across the Heads served by the same MSP. Constraint \eqref{eq:budget_constraint} next enforces local affordability by requiring the total cost incurred by MSP $m$ at time $t$ to remain within its currently available budget. Constraints \eqref{eq:bitrate_bounds} and \eqref{eq:framerate_bounds} bound the VE control variables within the admissible VRoom-specific service ranges, whereas \eqref{eq:accuracy_bounds} constrains the DT behavioral-fidelity control to a finite grid with step size $\Delta_\beta=0.05$. Finally, constraint \eqref{eq:budget_dynamics_noncoop} makes the inter-temporal coupling explicit by stating that expenditure at slot $t$ reduces the remaining budget available at slot $t{+}1$.}

{For tractability, we do not optimize every DT-related control variable. In particular, although an ideal allocation policy could also optimize the DT update frequency $\tau_h^{(t)}$ through $\lambda f_d$ and the structural-accuracy term $\varepsilon_h^{(t)}$, we assume these quantities are fixed in advance through service-level agreements between Heads and MSPs and therefore treat them as exogenous parameters.}

\subsection{Problem $\boldsymbol{\bar{P}}$: Cooperative setting}
While Problem $\boldsymbol{P}$ treated MSPs as isolated entities, which by design limits the collective potential of MSPs to cooperate and fulfill more requests particularly when MSPs operate with varying initial budgets and face fluctuating dynamic loads across their service areas. In that case, some MSPs might experience client surges and fail to meet demands, while others with minimal loads refrain from offering assistance. Hence, we now extend our problem $\boldsymbol{P}$ to a cooperative setting $\boldsymbol{\bar{P}}$ in which MSPs belong to the same administrative domain and can redistribute unused credits through a General Credit Pool (GCP). 

The GCP is a centralized resource-sharing system that enables MSPs to dynamically redistribute computing capacity through a deposit-withdraw framework. Using General Credits (GCs) as standardized units, each equivalent to a monetary value (e.g., 1 GC = 1 USD). MSPs supports their peers by purchasing GCs on behalf of other MSPs and depositing these GCs into the pool, which resource-constrained MSPs can later withdraw to acquire general computational resources (e.g., 1 GC = 2 CPU cores for 1 hour). This mechanism eliminates the complexities of direct credit transfers while allowing MSPs serving high-demand VRooms to secure additional capacity seamlessly. By supporting prepurchasing GCs ahead of time, the GCP also reduces on-demand costs and allocation latency.
The GCP supports different credit valuation mechanisms depending on the administrative relationship among MSPs. In homogeneous cooperation, all MSPs operate under a single organization (e.g., a large cloud provider managing multiple regional MSPs). In this case, General Credits (GCs) act purely as internal accounting units: if MSP A deposits 1 GC into the GCP, MSP B, belonging to the same organization, can later withdraw exactly 1 GC with no loss of value, since coordination costs and financial reconciliation are internalized by the same authority.

In heterogeneous cooperation, MSPs belong to different companies (e.g., MSP A from Company X and MSP B from Company Y) that voluntarily form a coalition to share resources during demand fluctuations. In this setting, GCs represent inter-organizational resource obligations rather than internal credits. Because resource contributions and withdrawals typically occur at different times and often under asymmetric demand conditions, the effective value of a GC may vary. For example, if MSP A contributes 1 GC during a low-demand period, MSP B may later withdraw resources during a peak-demand period, in which case the withdrawal may correspond to less than 1 GC of effective capacity (e.g., 0.95 GC), reflecting coordination overhead, delayed reciprocity, and scarcity-driven valuation. In $\boldsymbol{\bar{P}}$, we focus exclusively on the homogeneous cooperation model to isolate the impact of cooperative resource sharing without introducing inter-organizational pricing dynamics. Extending the GCP to fully dynamic, demand-aware credit pricing across heterogeneous MSP coalitions is a natural direction for future work.


{The main modeling choice in CIVIC for GCP donations is \emph{surplus-only donation}. Each MSP first determines its local VE/DT allocation and only afterward donates a fraction of the \emph{remaining} local surplus to the GCP. This ordering separates the local immersion-control decision from the inter-MSP coupling and yields a clearer pool-evolution model. The local initial condition remains $Budget_m^{(1)}=Budget_m^{(0)}$, while the pool starts empty, i.e., $GCP^{(1)}=0$. To express this mechanism compactly, we next define the local surplus and deficit quantities. 
For compactness, let $[x]_+ \overset{\mathrm{def}}{=} \max\{0,x\}$ and define}
{
\[
\begin{aligned}
S_m^{(t)} &\overset{\mathrm{def}}{=} \bigl[Budget_m^{(t)} - {Cost_{Total}^{m}}^{(t)}\bigr]_+, \\
W_m^{(t)} &\overset{\mathrm{def}}{=} \bigl[{Cost_{Total}^{m}}^{(t)} - Budget_m^{(t)}\bigr]_+.
\end{aligned}
\]
}
{where $S_m^{(t)}$ denotes the post-allocation local surplus of MSP $m$ and $W_m^{(t)}$ denotes its deficit before GCP support.}
{Using these auxiliary quantities, the cooperative optimization problem can be written compactly as follows:}
{
\begin{align}
\boldsymbol{\bar{P}}: \quad \max_{\{\mathbf{u}_h^{(t)}\}, \{\delta_m^{(t)}\}} \quad & \sum_{t=1}^{T} \sum_{m=1}^{M} L_m^{(t)} \label{eq:objective2}
\end{align}
}

{This formulation preserves the non-cooperative constraints of Problem $\boldsymbol{P}$ and augments them with the following cooperative constraints:}
\allowdisplaybreaks[4]
{
\begin{align*}
& 0 \leq \delta_m^{(t)} \leq 1,
\quad \forall m \in \mathcal{M}, \forall t \in \mathcal{T} \tag{c9}\label{eq:donation_bounds} \\
& D_m^{(t)} = \delta_m^{(t)} S_m^{(t)},
\quad \forall m \in \mathcal{M}, \forall t \in \mathcal{T} \tag{c10}\label{eq:donation_amount} \\
& \sum_{m=1}^{M} W_m^{(t)} \leq GCP^{(t)},
\quad \forall t \in \mathcal{T} \tag{c11}\label{eq:pool_capacity_constraint} \\
& 0 \leq GCP^{(t)} \leq \sum_{m=1}^{M} Budget_m^{(0)}, \quad \forall t \in \mathcal{T},
\qquad GCP^{(1)} = 0 \tag{c12}\label{eq:pool_non_negative_max} \\
& Budget_m^{(t+1)} = S_m^{(t)} - D_m^{(t)},
\quad \forall m \in \mathcal{M}, \forall t \in \mathcal{T}\setminus\{T\} \tag{c13}\label{eq:budget_dynamics_coop}\\
& GCP^{(t+1)} = GCP^{(t)} + \sum_{m=1}^{M} D_m^{(t)} - \sum_{m=1}^{M} W_m^{(t)},
\quad \forall t \in \mathcal{T}\setminus\{T\}
\tag{c14}\label{eq:pool_dynamics}
\end{align*}
}

{The cooperative constraints can likewise be read in sequence. Constraint \eqref{eq:donation_bounds} restricts the donation decision to a valid fraction in $[0,1]$. Constraint \eqref{eq:donation_amount} then converts that fraction into an actual donation amount by applying it only to the post-allocation local surplus $S_m^{(t)}$, so MSPs donate only after satisfying their own local spending decision. Constraint \eqref{eq:pool_capacity_constraint} enforces the core pool-feasibility rule, namely that the aggregate deficit requested from the GCP at any slot cannot exceed the credits currently available in the pool. Constraint \eqref{eq:pool_non_negative_max} bounds the GCP state and specifies the initial condition $GCP^{(1)}=0$. Constraint \eqref{eq:budget_dynamics_coop} then updates the remaining local budget by subtracting the donated amount from the current surplus. Finally, \eqref{eq:pool_dynamics} governs the shared GCP evolution by adding all donations and subtracting all deficits covered at the current slot. This explicit surplus/deficit accounting makes the cooperative coupling mathematically transparent while preserving the all-or-nothing fulfillment rule of $\boldsymbol{P}$.}

{Although credit sharing can substantially improve system-wide performance, suboptimal joint GCP-donation and allocation strategies may still degrade performance below non-cooperative levels, as later shown in Section \ref{sec:expeval}. Relative to Problem $\boldsymbol{P}$, Problem $\boldsymbol{\bar{P}}$ is more complex because of the additional decision variable $\delta_m^{(t)}$ and the inter-temporal coupling introduced by the GCP mechanism. This added complexity further motivates the use of advanced solution methods such as Deep Reinforcement Learning (DRL).}
\subsection{Problems complexity}

Both optimization problems $\boldsymbol{P}$ and $\boldsymbol{\bar{P}}$ are non-convex for two complementary reasons. First, their feasible sets are non-convex because several decision variables, namely $B_h^{(t)}$, $f_h^{(t)}$, and the grid-based $\beta_h^{(t)}$, are restricted to discrete admissible sets. Second, even under the standard continuous relaxation of these discrete controls to their enclosing intervals, the resulting objective/constraint functions remain non-convex. To show the latter point, we inspect the Hessians of representative continuous components. Recall that a twice-differentiable function is convex only if its Hessian matrix is positive semidefinite everywhere in its domain.

Firstly, the SSIM function in \eqref{eq:ssim} contains the term $B_h^{(t)-(e_3 + e_4\omega_h)}$ with negative exponent. The second derivative is:
\begin{align*}
\frac{\partial^2 \sigma_h^{(t)}}{\partial (B_h^{(t)})^2} &= -(e_1 + e_2\omega_h) \cdot (-(e_3 + e_4\omega_h)) \\ & \cdot (-(e_3 + e_4\omega_h) - 1) \cdot B_h^{(t)-(e_3 + e_4\omega_h)-2}
\end{align*}
Since $(e_3 + e_4\omega_h) > 0$ and $(e_1 + e_2\omega_h) > 0$ from the model coefficients, this second derivative is negative, proving non-convexity.

Secondly, the resource scaling function in \eqref{eqn:scaling} mixed partial derivative with respect to $f_h^{(t)}$ and $C_h^{(t)}$ in the Hessian is:
\begin{align*}
\frac{\partial^2 {R_s^y}^{(t)}}{\partial f_h^{(t)} \partial C_h^{(t)}} &= R_b^y \cdot \frac{\kappa}{f_{v_{\min}}^\kappa} \cdot f_h^{(t)\kappa-1} \cdot \vartheta \cdot C_h^{(t)\vartheta-1} \cdot g'(C_h^{(t)}) \\
&\quad + R_b^y \cdot \frac{\kappa}{f_{v_{\min}}^\kappa} \cdot f_h^{(t)\kappa-1} \cdot C_h^{(t)\vartheta} \cdot g''(C_h^{(t)})
\end{align*}
The term $g''(C_h^{(t)}) = -\iota \frac{\Lambda(\Lambda-1)}{V_{h_{\max}}^\Lambda} C_h^{(t)\Lambda-2}$ is negative when $0 < \Lambda < 1$, creating negative eigenvalues in the Hessian matrix and establishing non-convexity.

Thirdly, although $\beta_h^{(t)}$ is selected from a discrete grid in the actual formulation, its continuous relaxation within the interval $[\beta_{v_{\min}},\beta_{v_{\max}}]$ remains non-convex. In particular, the harmonic mean accuracy function in \eqref{eqn:dtacc} has second derivatives of the form:
\begin{align*}
\frac{\partial^2 A_h^{(t)}}{\partial (\beta_h^{(t)})^2} &= \frac{6\beta_h^{(t)-3}}{(\varepsilon_h^{(t)-1} + \beta_h^{(t)-1} + \tau_h^{(t)-1})^3} - \\ & \quad \frac{6\beta_h^{(t)-4}}{(\varepsilon_h^{(t)-1} + \beta_h^{(t)-1} + \tau_h^{(t)-1})^2}
\end{align*}
\updated{which yields negative values for typical parameter ranges, confirming non-convexity through negative Hessian eigenvalues. The mathematical complexity classifies both problems as non-convex mixed-integer nonlinear programming (MINLP) instances, rigorously justifying the DRL approach over traditional convex optimization methods.}

\section{THE PROPOSED DRL-BASED SOLUTIONS}
\label{sec:ourmodel_civic}

\begin{figure*}
    \centering
    \includegraphics[width=0.80\linewidth]{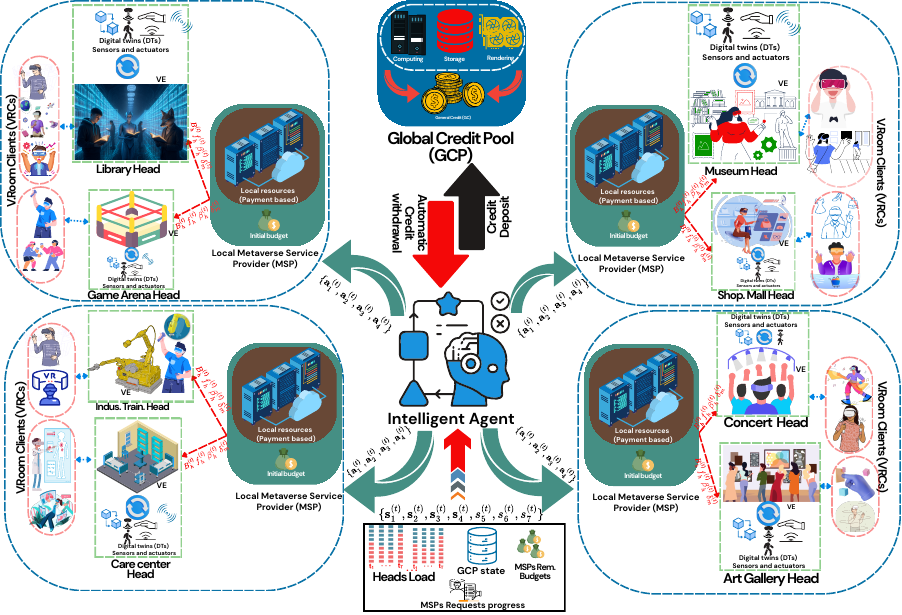}
    \caption{The proposed Cooperative Immersion Via Intelligent Credit-sharing in DRL-Powered Metaverse (CIVIC)}
    \label{fig:oursol}
\end{figure*}

The resource-allocation problems formulated in the previous section are non-convex MINLPs and present significant optimization challenges due to their NP-hard nature (see Appendix), dynamic environment characteristics (e.g., varying VRC counts from one timestep to another), long-term temporal coupling (i.e., allocation at $t$ affects the remaining budget at $t+1$), heterogeneous VRoom requirements, and the need for real-time decision making. The non-cooperative problem (Problem $\boldsymbol{P}$) requires each MSP to jointly optimize mixed decision variables $B_h^{(t)}$, $f_h^{(t)}$, and $\beta_h^{(t)}$ for multiple Heads while adhering to budget and immersion constraints. The cooperative problem (Problem $\boldsymbol{\bar{P}}$) adds a further coupling layer through the shared-pool donation variable $\delta_m^{(t)}$, requiring MSPs to balance local service quality with system-wide support. Traditional optimization solvers such as BARON face difficulty because of the system dynamics, real-time requirements, and the rapid growth in dimensionality as the number of MSPs and Heads increases. 

Therefore, to tackle these challenges, we develop DRL‑based solutions for both $\boldsymbol{P}$ and $\boldsymbol{\bar{P}}$, leveraging a centralized DRL framework capable of learning optimal resource‑allocation and cooperation policies through continuous interaction with the environment. In the following sections, we first introduce the DRL framework and the selected algorithm. Then, in the subsequent subsections, we start by formulating $\boldsymbol{P}$ as a Markov Decision Process (MDP) and evaluating the resulting agent through a comprehensive set of experiments. We then extend this approach to the cooperative setting, formulating $\boldsymbol{\bar{P}}$ as an MDP that not only allocates resources across individual MSPs and their associated Heads but also facilitates inter‑MSP cooperation through the GCP, as illustrated in \fref{fig:oursol}.

\subsection{DRL algorithm}

In contrast to traditional machine and deep learning solutions, which often require building knowledge based on a training dataset to make predictions of values or classes, Reinforcement Learning (RL) and its deep learning-assisted counterpart, DRL, take a different approach and serve a different purpose. DRL is often recognized as an advanced heuristic solution that specializes in solving problems requiring long-term optimization goals by optimally selecting sequential actions under uncertainty. For that to work, DRL algorithms collect millions and even billions of interactions with the simulated environment.
DRL is theoretically grounded in Markov Decision Process (MDP), which is defined as $\langle \mathbb{S}, \mathbb{A}, \mathbb{P}, \mathbb{R}, {\Gamma} \rangle$. In this framework, the agent continuously observes the environment state $ \mathbb {S} $, selects an action $ \mathbb {A} $, and receives a reward $ \mathbb {R} $ discounted by $ \Gamma $, along with the updated state $ \mathbb {S'} $. The transition from $\mathbb{S}$ to $\mathbb{S'}$ occurs probabilistically, as defined by $\mathbb{P}$. After training the agent across multiple episodes, it learns a policy $\pi$, effectively associating states with actions that maximize cumulative rewards. 
Among the different DRL algorithms, we selected Proximal Policy Optimization (PPO) \cite{schulman2017proximal}. PPO offers several advantages that make it our optimal choice. First, it provides stable low-dimensional control outputs that can be projected onto the admissible bitrate, frame-rate, and behavioral-fidelity grids used by our environment. Moreover, it incorporates a clipped surrogate objective that balances exploration and exploitation, preventing destructive policy updates and leading to stable and reliable training. PPO is also highly amenable to parallelization, allowing multiple workers or environments to collect experience simultaneously, which significantly accelerates training and improves sample diversity. Additionally, PPO is recognized for its sample efficiency and ease of implementation, while achieving comparable or superior performance in practice.

In what follows, we will present the formulation of our problem as an MDP by defining the environment state, actions, and rewards.
\subsection{Non-Cooperative}
Starting with our non-cooperative scenario, we propose a centralized agent that tunes the various parameters of all Heads belonging to all MSPs. We begin by introducing the state space, followed by the agent's action and reward, which are used to guide the agent. In our MDP formulation, a vector is presented by a bold letter.
The state space, denoted as $\mathbb{S}$ is described as:
\begin{equation}
\mathbb{S} =\{\textbf{s}_1^{(t)}, \textbf{s}_2^{(t)}, \textbf{s}_3^{(t)}, \textbf{s}_4^{(t)}, s_5^{(t)}\}
\end{equation}
where $\textbf{s}_1^{(t)} =\{ \frac{{Budget_{m}^{(t)}}}{{Budget_{m}^{(0)}}} : \forall m \in \mathcal{M} \}$ conveys the set of remaining budget percentage per each MSP, $
\textbf{s}_2^{(t)} = \{\frac{\sum_{t'=1}^{t'=t} L_m^{(t')}}{|\mathcal{N}_m|} : \forall m \in \mathcal{M}\}$ signifies the percentage of requests fulfilled so far by each MSP and $|\mathcal{N}_m|$ is the count of requests needed to be fulfilled by $m$, 
$\textbf{s}_3^{(t)} = \{\frac{C_h^{(t)}}{V_{h_{max}}} : \forall h \in m, \forall m \in \mathcal{M} \}$ describes the number of clients in each VRoom as a percentage of VRoom's maximum capacity. 
$\textbf{s}_4^{(t)} = \{\mathbb{I}_{\text{request}}(h,t): \forall h \in \mathcal{H}\}$ shows the current active request indicators and finally $s_5^{(t)} = \frac{t}{T}$ indicates a progress till the episode ends.

Similarly, the action space denoted as $\mathbb{A}$, can be described as:  
\begin{equation}  
\mathbb{A} = \{ \textbf{a}_1^{(t)}, \textbf{a}_2^{(t)}, \textbf{a}_3^{(t)} \}  
\end{equation}  
where:  $ \textbf{a}_1^{(t)} = \{ B_h^{(t)} : \forall h \in \mathcal{H}_m, \forall m \in \mathcal{M} \} $ represents the bitrate allocation for each $ h $,  
$ \textbf{a}_2^{(t)} = \{ f_h^{(t)} : \forall h \in \mathcal{H}_m, \forall m \in \mathcal{M} \} $ signifies the frame rate assignments for each $ h $, and $\textbf{a}_3^{(t)} = \{ \beta_h^{(t)} : \forall h \in \mathcal{H}_m, \forall m \in \mathcal{M} \} $ indicates the behavioral-accuracy level selected for the DT in $h$'s VRoom. In implementation, the PPO output corresponding to $\beta_h^{(t)}$ is projected onto the nearest admissible level in the discrete grid $\mathcal{B}_{v_h}$ before being applied to the environment.

Since $\boldsymbol{P}$ focuses on maximizing the number of \textit{fulfilled} requests (i.e., requests with $I_h^{(t)} \ge I_{\text{threshold}}$) under constrained MSP budgets, the reward function at time $t$, denoted by $\mathbb{R}_t$, is designed to reward efficient threshold-meeting allocations while penalizing \textit{failed} requests that arise either from infeasible actions or from budget depletion that prevents service.

Therefore, we formulate $\mathbb{R}_t$ as:
\updated{\begin{align}
\mathbb{R}_t &= 
w_{\text{imm}} \cdot 
\underbrace{\sum_{m \in \mathcal{M}} \sum_{h \in \mathcal{H}_{m,active}^{(t)}} \phi(I_h^{(t)})}_{\text{Immediate reward at ($t \neq T$)}} \notag \ + \\
&\quad  
\underbrace{w_{\text{fin}} \cdot \sum_{t=1}^{T} \sum_{m \in \mathcal{M}} \sum_{h \in \mathcal{H}_{m,active}^{(t)}} \mathbf{1}[{0 < I_h^{(t)} \leq I_{\text{max}}}]}_{\text{Terminal reward at ($t = T$)}}
\label{eq:rt}
\end{align}}

where the first term in \eqref{eq:rt} represents the immediate reward given at anytime ($t \neq T$) and is the allocation efficiency score of all MSPs and is calculated by aggregating the output of  the piece-wise efficiency function at $\phi(.)$ across all of the active Heads Immersions' scores at $t$ as:
\[
\smallmath[0.75]{\phi(I_h^{(t)}) = \begin{cases}
1.5 & \text{if } I_{\text{threshold}} \leq I_h^{(t)} \leq 110\% \cdot I_{\text{threshold}} \\
0.5 - \min\left(0.3, \frac{I_h^{(t)} - I_{\text{threshold}}}{I_{\text{threshold}}}\right) & \text{if } I_h^{(t)} > 110\% \cdot I_{\text{threshold}} \\
-1 & \text{if } I_h^{(t)} < I_{\text{threshold}}\\
\end{cases}}
\]
Here, the $\phi(I_h^{(t)})$ outputs a higher score if the attained immersion is higher than the immersion threshold by $10\%$ at the max, while reducing the efficiency score as this over-allocation increases.
\updated{On the other hand, the second term in \eqref{eq:rt} represents the terminal reward which encourages the agent to gradually improve the number of satisfied requests by first teaching it to do as much requests as possible (under or over immersion allocation) and avoid MSPs failing and is calculated by aggregating the binary function output across all requests done in the whole episode as $\mathbf{1}[{0 < I_h^{(t)} \leq I_{\text{max}}}]$. Both immediate and terminal rewards are multiplied respectively by weights $w_{imm} \ \text{and} \ w_{fin}$, which are chosen experimentally.}

\subsection{Cooperative MSPs}
The cooperative problem $\bar{P}$ introduces significant complexity compared to the non-cooperative setting, requiring careful consideration of the appropriate solution methodology. While distributed multi-agent reinforcement learning approaches such as Centralized Training with Decentralized Execution (CTDE) offer advantages in terms of scalability and agent independence, they are not well-suited for our cooperative scenario due to several key factors.
First, the cooperative setting requires global awareness for optimal decision-making. Each MSP must monitor not only its local state but also the GCP state, peer MSP demands, and resource depletion patterns across the entire system. This global information is essential for MSPs to make informed decisions about resource allocation timing, donation strategies, and withdrawal patterns to optimize their ability to fulfill their total request quotas $\mathcal{N}_m$.
Second, the GCP mechanism necessitates centralized oversight to enforce global constraints. Critical system-wide constraints such as $\sum_{m=1}^{M} W_m^{(t)} \leq GCP^{(t)}$ require global coordination to ensure system feasibility. Additionally, optimizing donation timing across multiple MSPs requires synchronized behavior that is difficult to achieve through independent agents.
Third, our unified administrative control structure eliminates competitive dynamics typically present in multi-agent systems. Since all MSPs operate under the same administrative entity with aligned objectives, the system optimization naturally focuses on the true global optimization problem rather than seeking Nash equilibria among competing agents.
Finally, the manageable scale of our system (3-7 MSPs) makes centralized control computationally feasible while ensuring optimal coordination. The central administrative entity that manages MSPs can deploy and operate the centralized DRL model, enabling comprehensive control over resource orchestration and strategic cooperation.

Building upon the non-cooperative state space $\mathbb{S}$, the cooperative environment requires additional state information to capture the dynamics of the GCP and system coordination status. The extended state space $\bar{\mathbb{S}}$ incorporates two critical components: $s_6^{(t)} $, which indicates the available budget at time $t$ for GCP (i.e., $s_6^{(t)}=GCP^{(t)}$), and $s_7^{(t)}$, a stuck counter variable used to determine whether the system has entered a halting state (i.e., a state where no MSP needs to complete additional requests because all requests are finished, or the budget has been exhausted, and the GCP budget is depleted or insufficient to support any of the MSPs).

Thus, the updated state is defined as:

\begin{equation}
\bar{\mathbb{S}} =\{\textbf{s}_1^{(t)}, \textbf{s}_2^{(t)}, \textbf{s}_3^{(t)}, \textbf{s}_4^{(t)}, s_5^{(t)},s_6^{(t)}, s_7^{(t)}\}
\end{equation}

Moreover, the action space adds another action vector $\textbf{a}_4^{(t)}$, where $ \textbf{a}_4^{(t)} = \{ \delta_m^{(t)} :  \forall m \in \mathcal{M} \}$ is the donation fraction applied to each MSP's post-allocation surplus at time $t$.
Therefore 
\begin{equation}
\bar{\mathbb{A}} = \{ \textbf{a}_1^{(t)}, \textbf{a}_2^{(t)}, \textbf{a}_3^{(t)}, \textbf{a}_4^{(t)}\}
\end{equation}

\updated{We note that the choice of the donation vector $\textbf{a}_4^{(t)}$ directly affects the residual local budget through \eqref{eq:budget_dynamics_coop} and the shared cooperative state through \eqref{eq:pool_dynamics}. Hence, donation decisions influence future request feasibility even when they do not immediately alter the current slot's immersion values.}
The reward function remained the same, encouraging the satisfaction of more requests.

\begin{figure*}
    \centering
    \includegraphics[width=\textwidth]{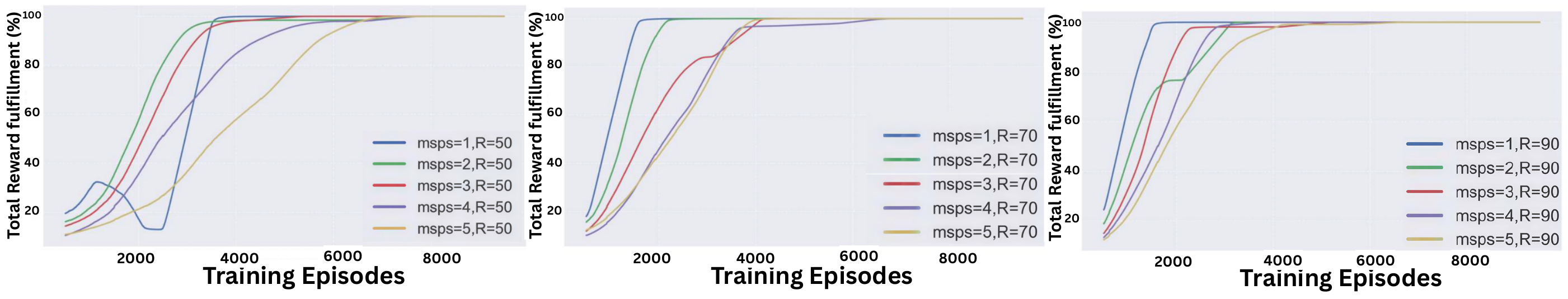}  
    \caption{\updated{Reward Convergence for non-GCP (non-cooperative) DRL solution, while increasing number of requests and number of MSPs}}
    \label{fig:non-cooperative_conv}
\end{figure*}

\begin{figure*}
    \centering
    \includegraphics[width=\linewidth]{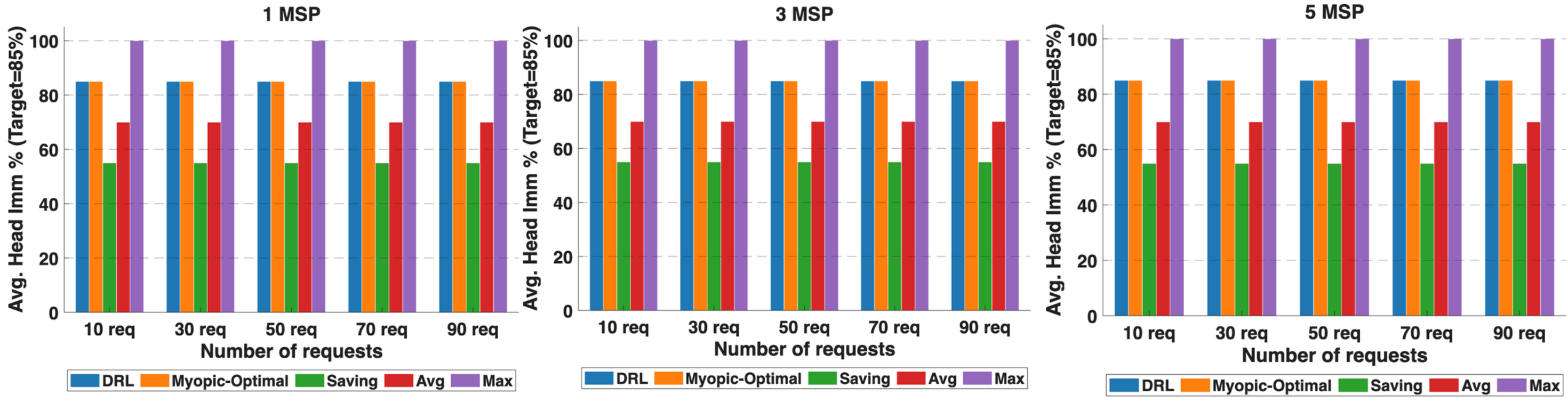} 
    \caption{\updated{Comparing our non-GCP (non-cooperative) DRL-based solution to other Baselines}}
    \label{fig:compsubfig4}
\end{figure*}

\begin{figure}[ht!]
    \centering
    \begin{subfigure}[b]{0.4\textwidth}
        \centering
        \includegraphics[width=\textwidth]{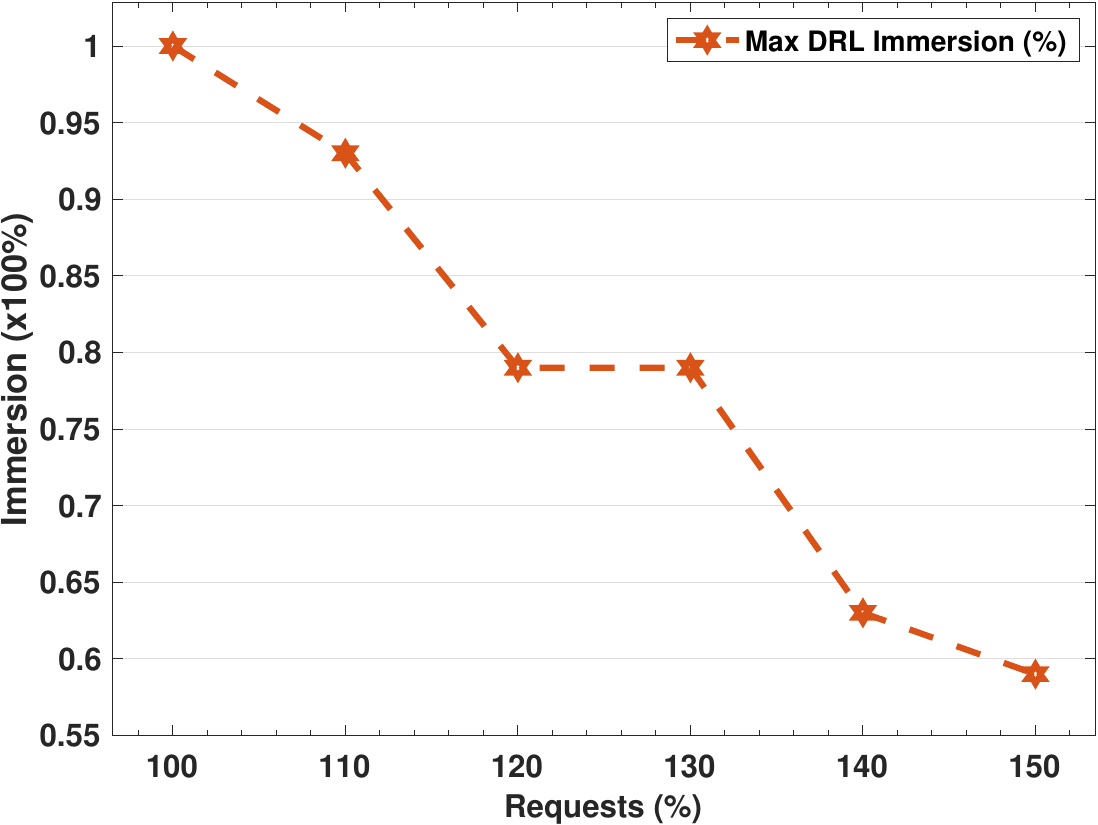} 
        \caption{\updated{Fixing budget while increasing number of requests}}
        \label{fig:exPsubfig1}
    \end{subfigure}
    \begin{subfigure}[b]{0.4\textwidth}
        \centering
        \includegraphics[width=\textwidth]{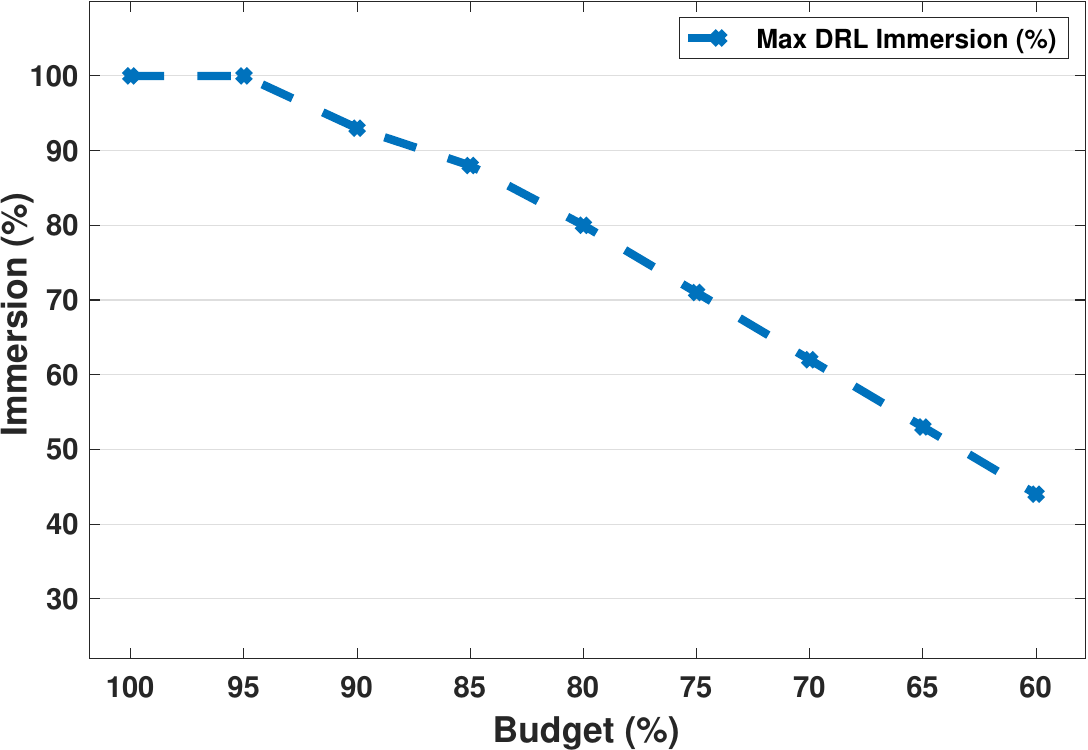} 
        \caption{\updated{Fixing Number of requests while reducing the budget}}
        \label{fig:exp2subfig2}
    \end{subfigure}
    \caption{\updated{Assessing the DRL-based Solution tolerance}}
    \label{fig:non-cooperative-asse}
\end{figure}

\section{PERFORMANCE EVALUATION \& DISCUSSION}
\label{sec:expeval}
In this section, we present a comprehensive evaluation of our proposed solutions for both $\boldsymbol{P}$ and $\boldsymbol{\bar{P}}$. We begin by describing our experimental setup, followed by a detailed analysis of the performance characteristics and comparative evaluation against baseline approaches.

Our evaluation environment for the non-cooperative setting ($\boldsymbol{P}$) comprises multiple MSPs serving different Heads, each hosting various VRoom types with heterogeneous resource requirements. We consider an ideal range of 1-5 MSPs, where each MSP manages one Head with VRooms types as specified in \tref{tab:generalparams}. Each MSP is tasked with fulfilling a predetermined number of fixed requests, unified across all MSPs ($\mathcal{N}_m \in \{50,70,90\}$). Additionally, the immersion threshold, $I_{\text{threshold}}$, was unified across all MSPs for consistency across experiments. Furthermore, MSPs' budgets are equally distributed between MSPs to simulate independent operation without collaboration, and the number of VRCs in VRooms was fixed across timesteps.
On the other hand, the cooperative setting ($\boldsymbol{\bar{P}}$) implements more challenging and realistic operational conditions, reflecting collaborative scenarios with a higher number of MSPs (i.e., 3-7 MSPs) and dynamic requests.
Moreover, in this cooperative environment, a more challenging and unbalanced budget scenario is used to simulate realistic operational conditions, where MSPs must collaborate under resource constraints. 

\updated{Furthermore, to model dynamic user participation, we initialize each VRoom with a non-zero number of active VRCs and subsequently employ a double-Poisson arrival–departure process. At each timestep, the number of arriving and departing VRCs for each VRoom is independently drawn from Poisson distributions with rates $r_{arr}=0.4$ and $r_{dep}=0.7$ respectively. This results in a time-varying room occupancy that captures natural user movement across VRooms and enables the evaluation of cooperative resource management under dynamic and fluctuating demand conditions. Other traffic generation models (e.g., bursty, time-correlated, or event-driven processes), as well as real-world datasets capturing attendance dynamics of online events such as virtual concerts or live sports streams, can be readily incorporated without modifying the proposed optimization or learning framework.}

\subsection{Non-cooperative setting ($\boldsymbol{P}$)}
After discussing the experimental setup, we assess the performance and efficiency of our proposed DRL solution for problem $\boldsymbol{P}$.

Initially, we analyzed the learning behavior and convergence characteristics of the DRL agent trained with a unified minimum immersion threshold of 85\%, as presented in \fref{fig:non-cooperative_conv}. The same figure illustrates the total reward convergence across all MSPs, with convergence curves indicating an increase in difficulty as both the number of MSPs and requests increase. 
Following the agent's acquisition of the optimal policy, we conducted three experiments to benchmark its performance.

Firstly, we compare our DRL-based solution against several baseline approaches. Our baseline comparisons includes multiple resource allocation policies such as: \textit{Saving}, which demonstrate a cost-conserving policy that allocates minimal resources to maximize the number of fulfilled requests, the \textit{Average}, which consistently selects average resource allocations for all requests, the \textit{Max} policy, which invariably assigns maximum resources to each request. \updated{The \textit{Myopic-Optimal} solution employs a computationally expensive greedy optimization approach that makes locally optimal decisions at each timestep. Specifically, at each time $t$, it:
(1) Evaluates all feasible resource allocations ($B_h^{(t)}$, $f_h^{(t)}$, $\beta_h^{(t)}$) that satisfy constraints \eqref{eq:active_heads}--\eqref{eq:budget_dynamics_noncoop} for all Heads,
(2) Selects the allocation that achieves the immersion threshold $I_h^{(t)} \geq I_{\text{threshold}}$,
(3) Among solutions that meet the threshold, choose the one with the minimum cost to save budget.}
The results in Figure \ref{fig:compsubfig4} present the immersion levels achieved by each solution across configurations of MSPs with request volumes. The analysis reveals that our DRL solution consistently tracks the optimal allocation policy across all experimental conditions. In contrast, the saving policy under-provisions resources, yielding suboptimal immersion levels, while the average policy shows marginal improvement over the saving approach but still falls short of adequate provisioning. Conversely, the max policy over-provisions resources, resulting in unnecessarily high immersion levels that would adversely impact MSPs' budgets in the long term.
\begin{figure}
    \centering
    \begin{subfigure}[b]{0.45\textwidth}
        \centering
        \includegraphics[width=\textwidth]{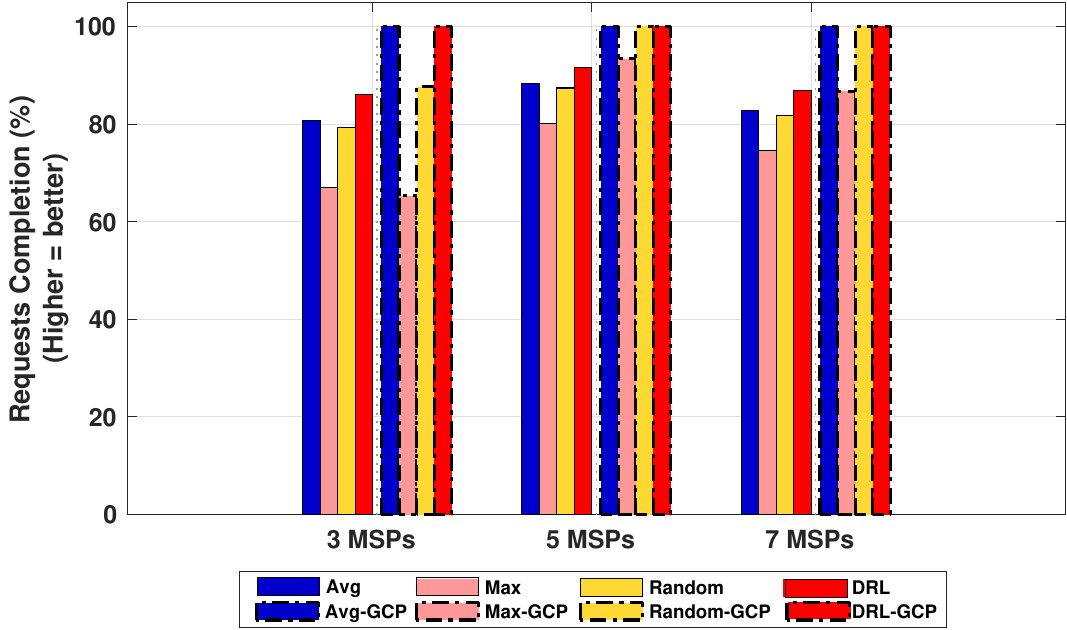} 
        \caption{100 Timesteps}
    \end{subfigure}
     \begin{subfigure}[b]{0.45\textwidth}
        \centering
        \includegraphics[width=\textwidth]{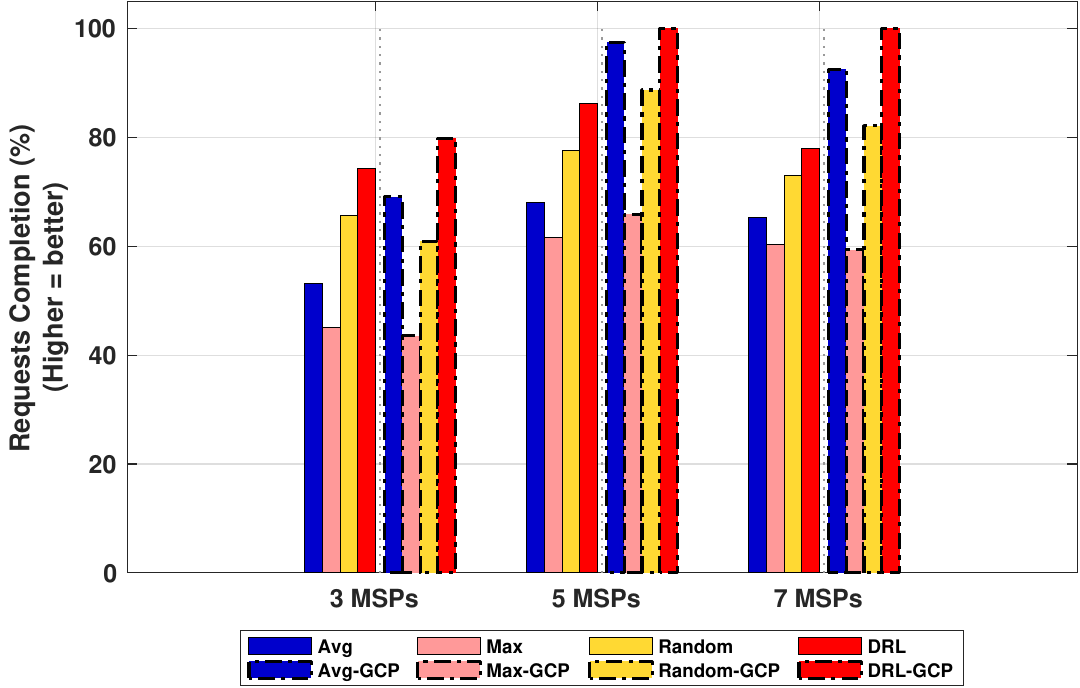} 
        \caption{150 Timesteps }
    \end{subfigure}
    \begin{subfigure}[b]{0.45\textwidth}
        \centering
        \includegraphics[width=\textwidth]{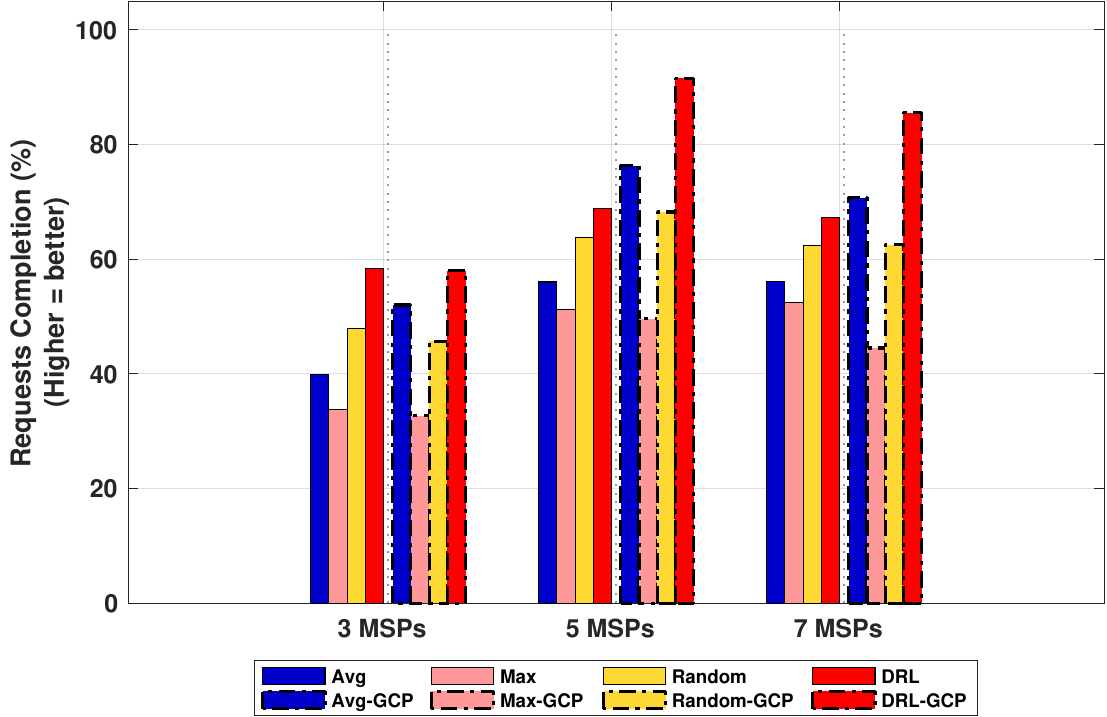}
        \caption{200 Timesteps}
    \end{subfigure}
    \caption{Assessments of the Cooperative Control of the MSPs}
    \label{fig:requests-completion}
\end{figure}

Secondly, \fref{fig:non-cooperative-asse} presents our second and third experiments. In \fref{fig:exPsubfig1} we present our second experiment, which evaluates the system's ability to manage resource allocation as the number of requests increases under a fixed budget. Initially, the agent was trained to provide 100\% immersion for all requests, and we systematically increased the request volume by 10\% in each subsequent trial. The results show that our agent demonstrates adaptive behavior, successfully handling up to 150\% of the original request load while maintaining an average immersion level of 60\%. This indicates that the system can gracefully manage increased demand by reducing service quality rather than failing. In \fref{fig:exp2subfig2}, we present our third experiment, where we assess the system's resilience in accepting all given requests with the highest possible immersion level while reducing the budget in each trial, simulating a poor initial budget allocated to the MSPs. In that experiment, the system was initially trained to achieve 100\% of immersion for all of the requests, and the total price was recorded before being progressively reduced by 5\%. The results show that the agent was able to continue serving all of the requests, but with a sublinear drop in immersion value.

\begin{figure*}
    \centering
    \begin{subfigure}[b]{0.32\textwidth}
        \centering
        \includegraphics[width=\textwidth]{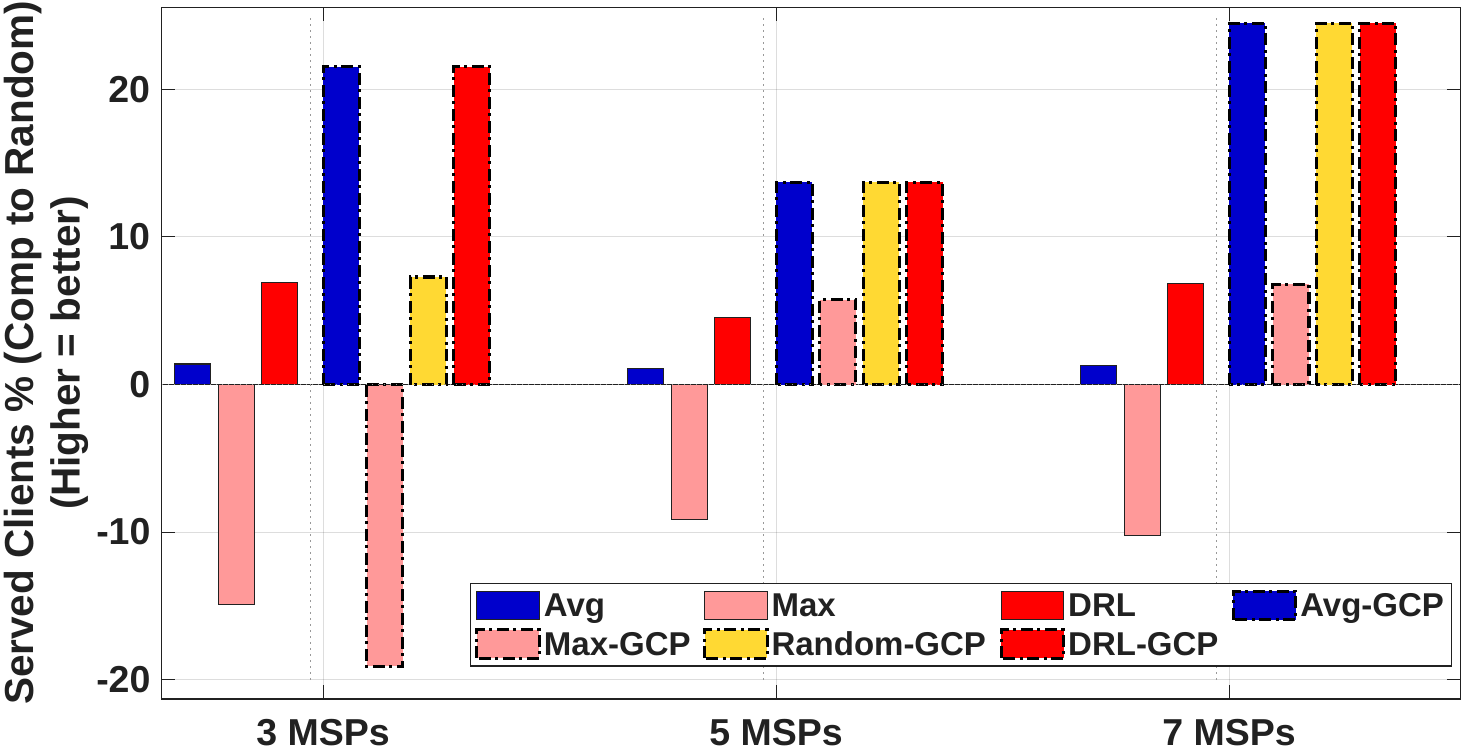} 
        \caption{\updated{100 Timesteps}}
    \end{subfigure}
     \begin{subfigure}[b]{0.32\textwidth}
        \centering
        \includegraphics[width=\textwidth]{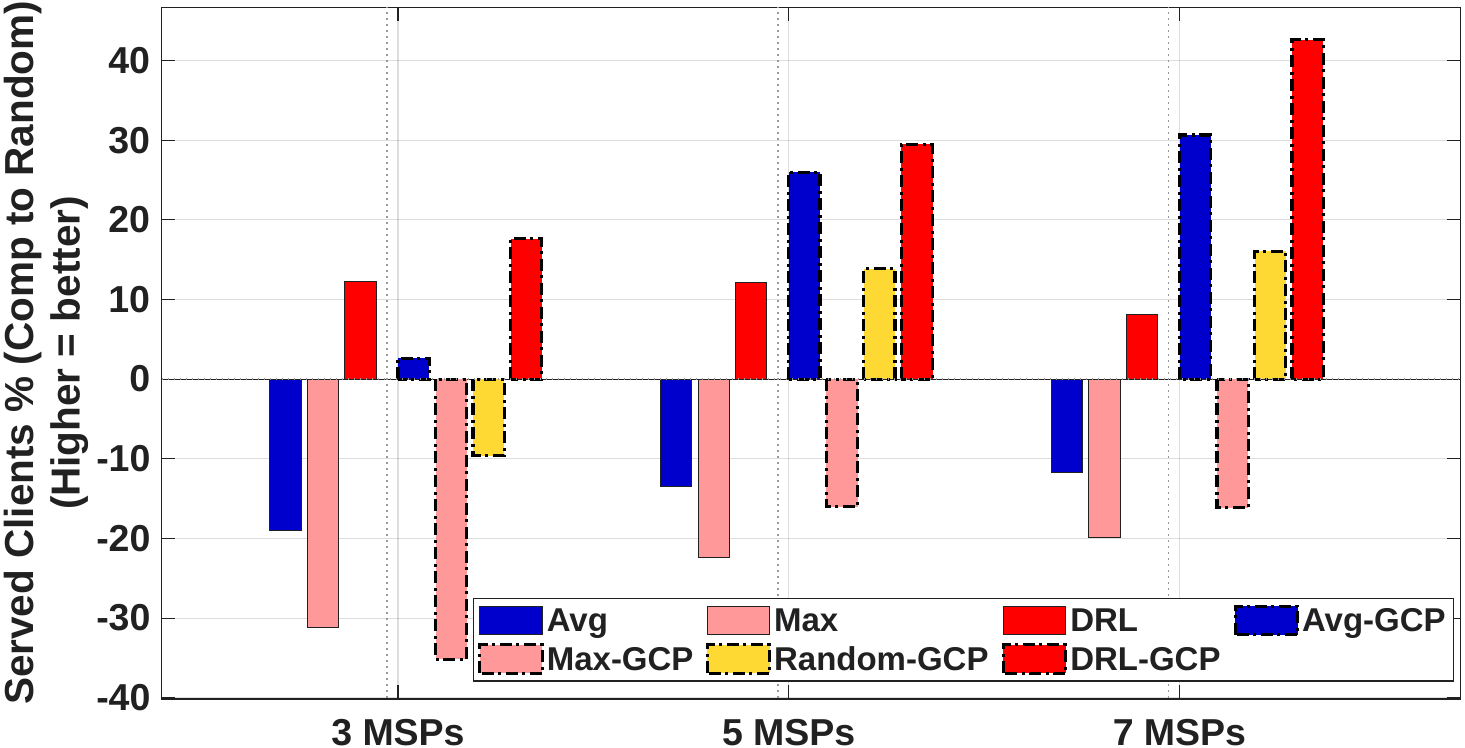} 
        \caption{\updated{ 150 Timesteps}}
    \end{subfigure}
    \begin{subfigure}[b]{0.32\textwidth}
        \centering
        \includegraphics[width=\textwidth]{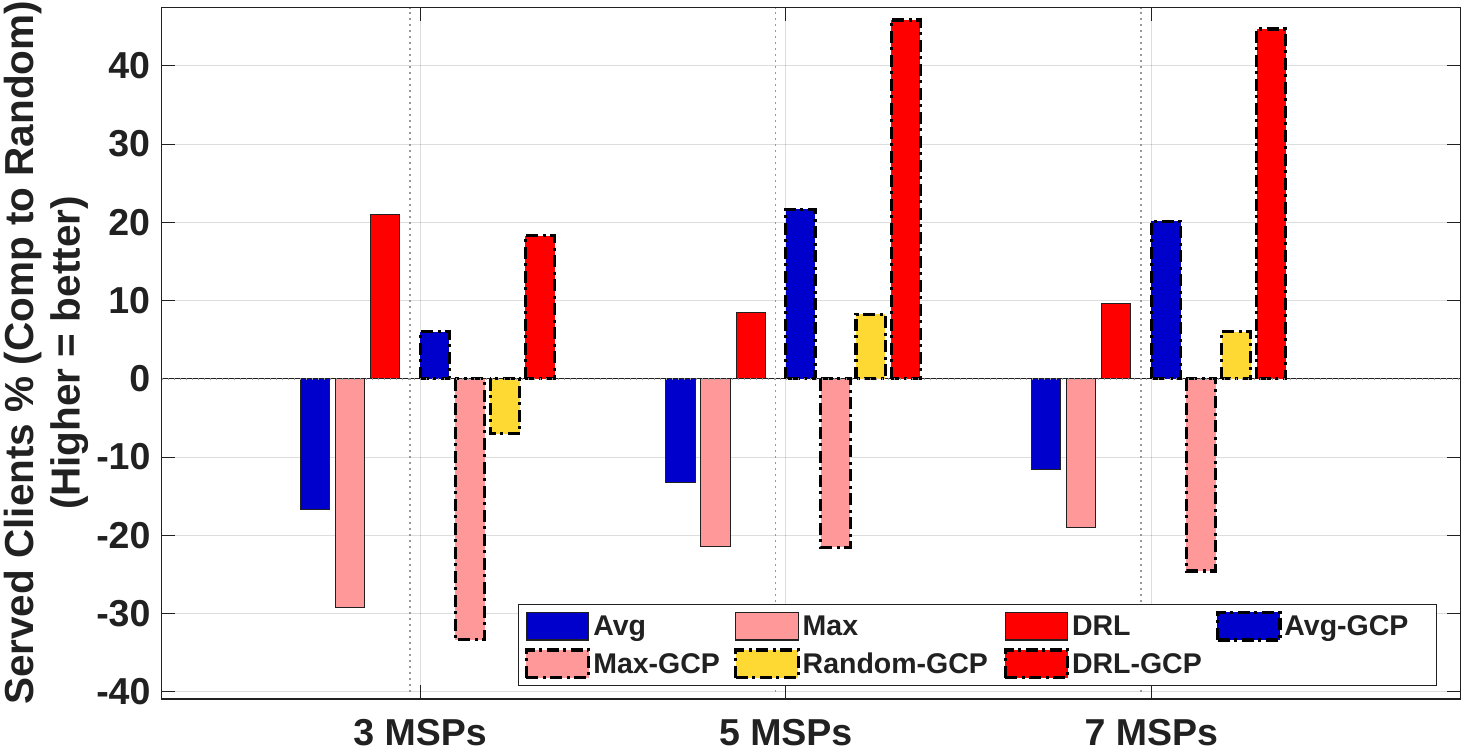} 
        
        \caption{\updated{200 Timesteps}}
    \end{subfigure}
    \caption{\updated{Comparison of percentage of Served clients compared to Random}}
    \label{fig:cop-serv}
\end{figure*}

\begin{figure*}
    \centering
    \begin{subfigure}[b]{0.32\textwidth}
        \centering
        \includegraphics[width=\textwidth]{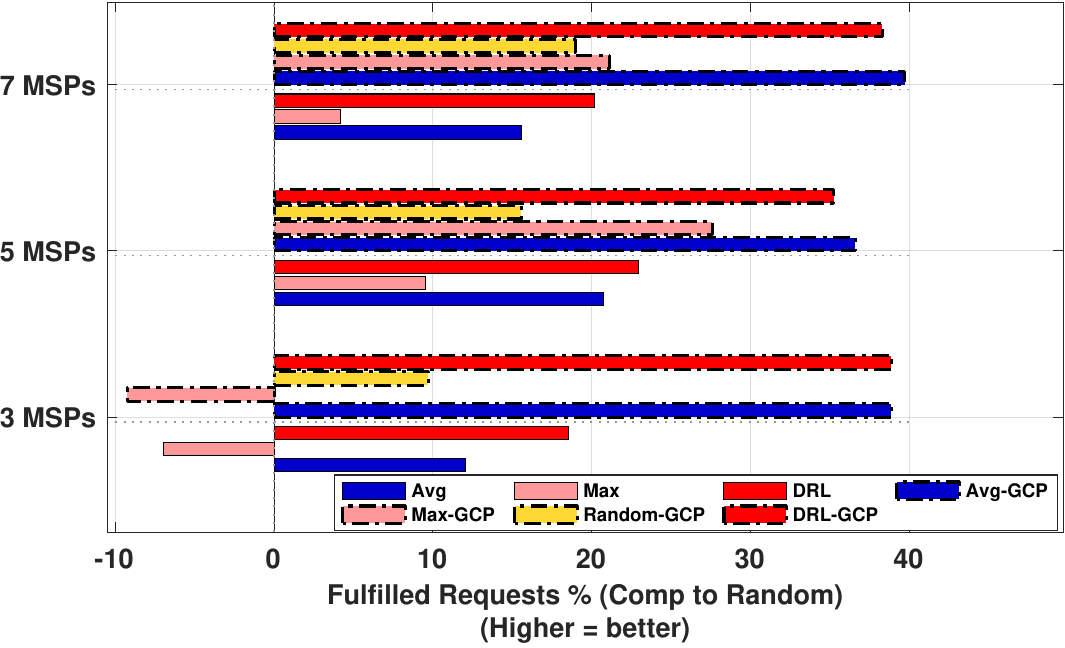} 
        \caption{\updated{100 Timesteps}}
    \end{subfigure}
     \begin{subfigure}[b]{0.32\textwidth}
        \centering
        \includegraphics[width=\textwidth]{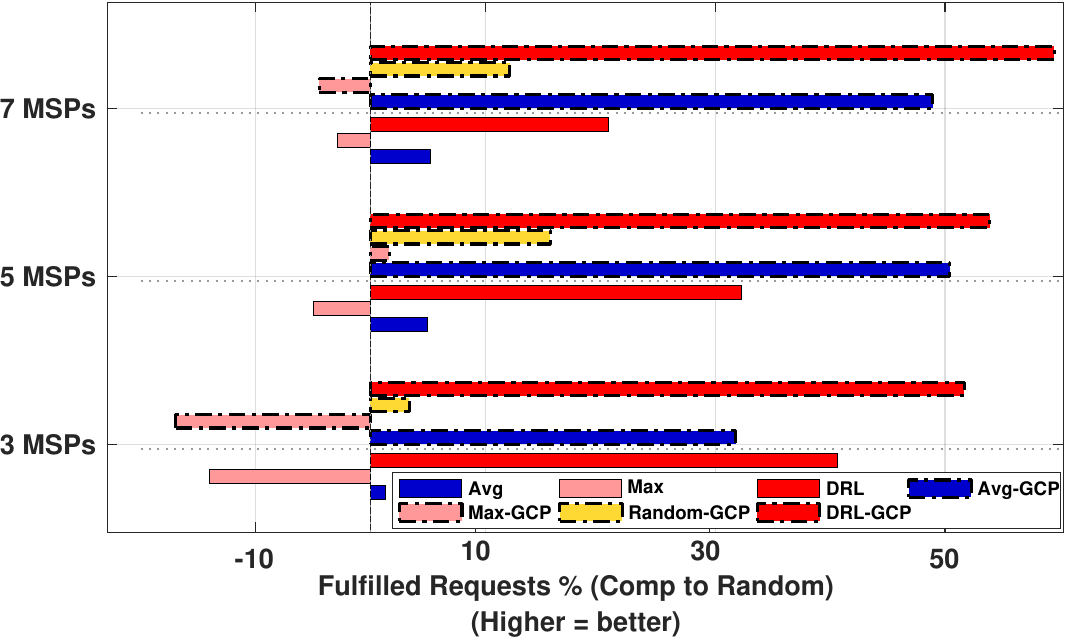} 
        \caption{\updated{150 Timesteps}}
    \end{subfigure}
    \begin{subfigure}[b]{0.32\textwidth}
        \centering
        \includegraphics[width=\textwidth]{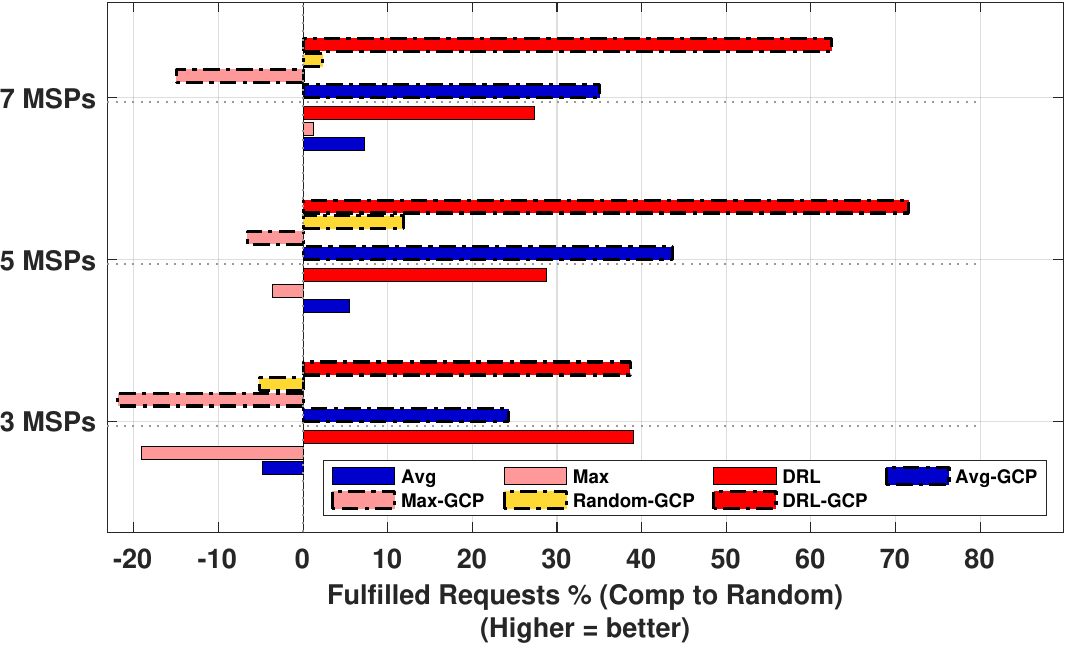} 
        \caption{\updated{200 Timesteps}}
    \end{subfigure}
    \caption{\updated{Comparison of percentage of fulfilled requests compared to Random }}
    \label{fig:cop-ful}
\end{figure*}
\subsection{\updated{Cooperative setting ($\boldsymbol{\bar{P}}$)}}

After assessing our DRL-based solution solving problem $\boldsymbol{P}$, we outline the different experiments used to assess our proposed solution to the cooperative setting $\boldsymbol{\bar{P}}$. To better assess our DRL with GCP solution (Hereafter, referred to as DRL-GCP), we benchmark its performance against two sets of solutions: the first set includes some of the non-GCP baselines introduced in the previous sections (i.e., Average, Max, and Random), while the second set comprises their improved counterparts with our GCP mechanism, such as Max-GCP, Avg-GCP, and Random-GCP. The second set of policies works with their intended behaviour while considering GCP dynamics. For example, Max-GCP policy dictates 1) the allocation of the maximum resources for all Heads (as in non-GCP Max policy) and 2) depositing the maximum possible credit to the GCP every timestep. 
Moreover, if no resource credit is available at any MSP to serve its Heads, the required resource credits for that MSP are automatically withdrawn from the GCP, and no resource deposit is made from that MSP at that timestep.

In this section, we create a total of three experiments as follows: 
\subsubsection{Experiment 1: Evaluating Resource allocation}


\begin{figure*}
    \centering
    \begin{subfigure}[b]{0.33\textwidth}
        \centering
        \includegraphics[width=\textwidth]{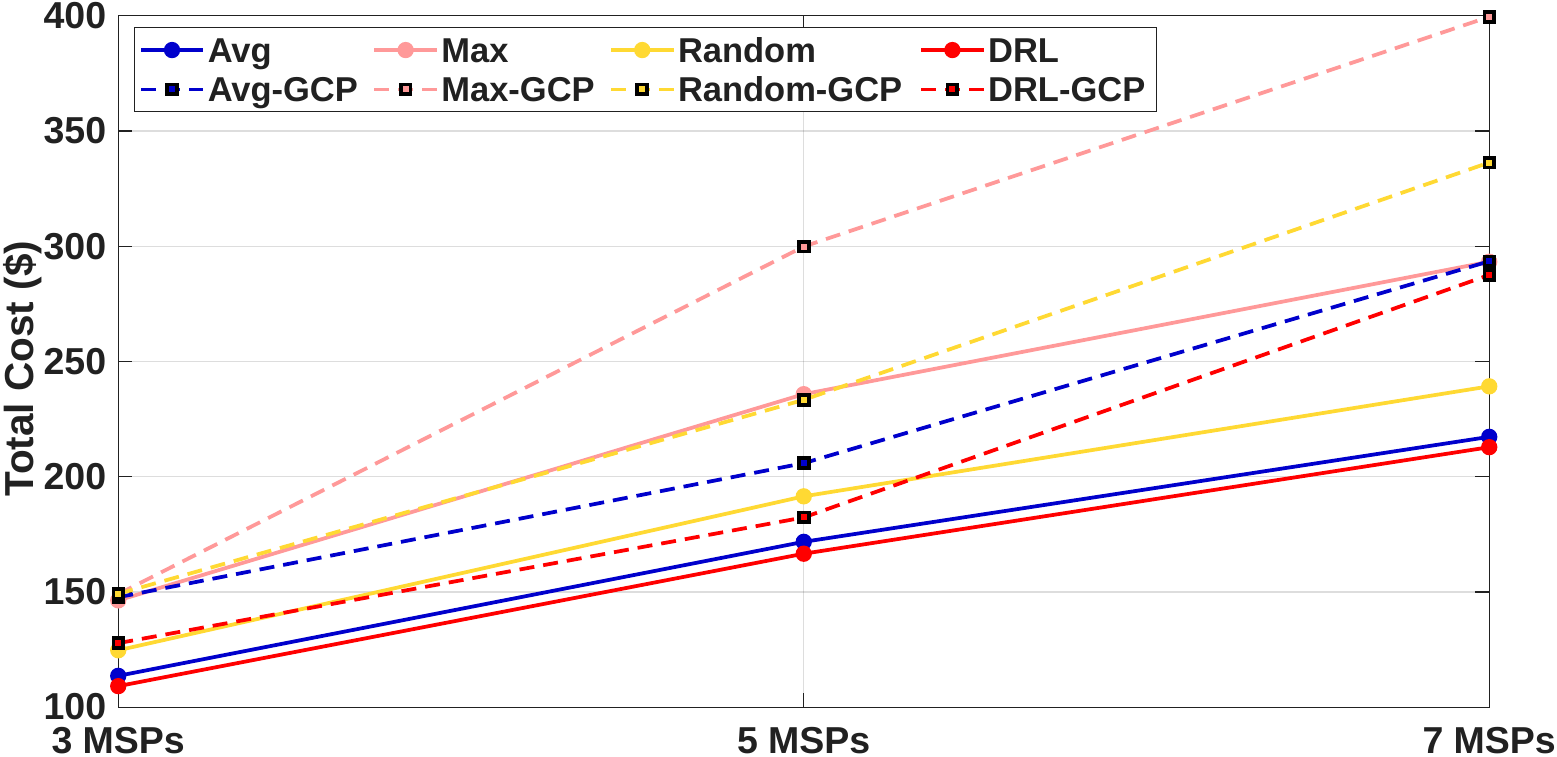} 
        \caption{\updated{100 Timesteps}}
    \end{subfigure}
     \begin{subfigure}[b]{0.32\textwidth}
        \centering
        \includegraphics[width=\textwidth]{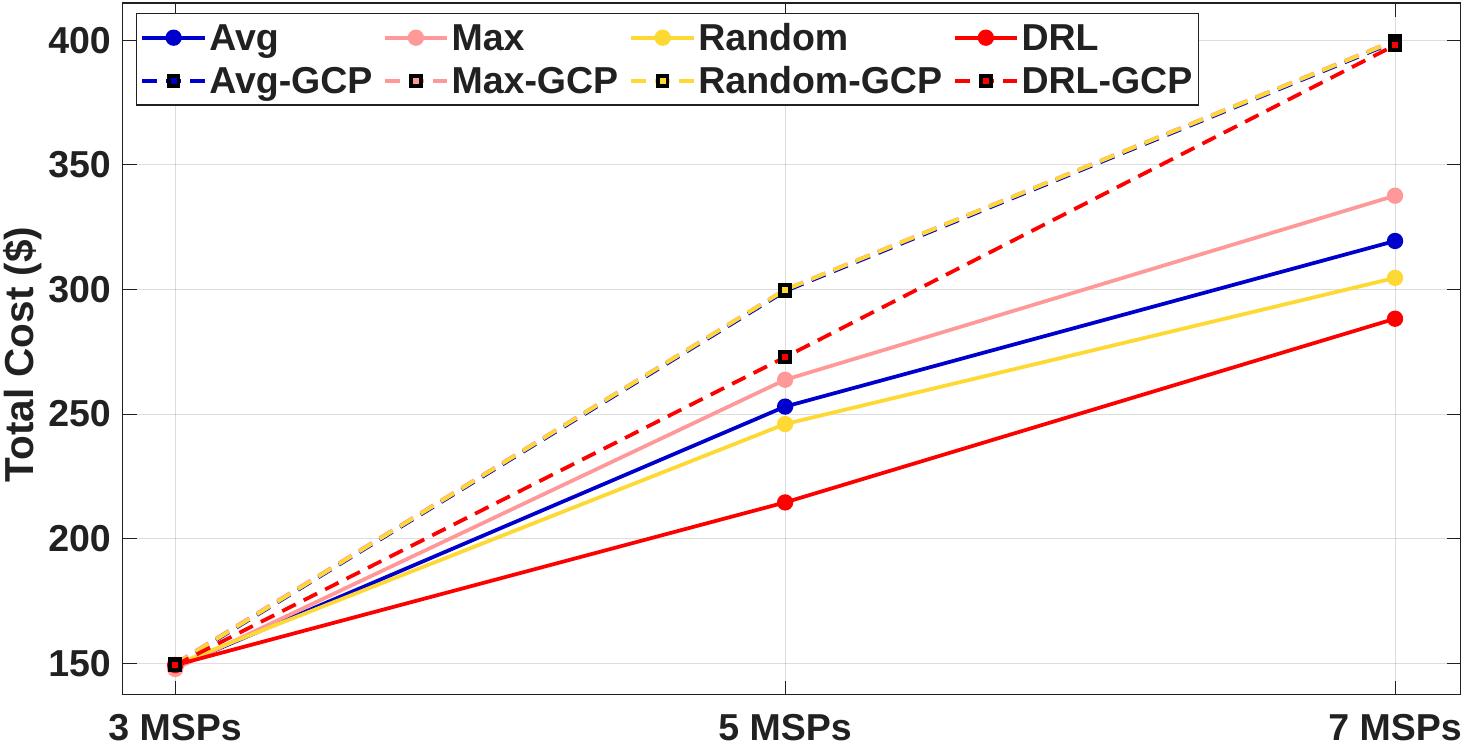} 
        \caption{\updated{150 Timesteps }}
    \end{subfigure}
    \begin{subfigure}[b]{0.32\textwidth}
        \centering
        \includegraphics[width=\textwidth]{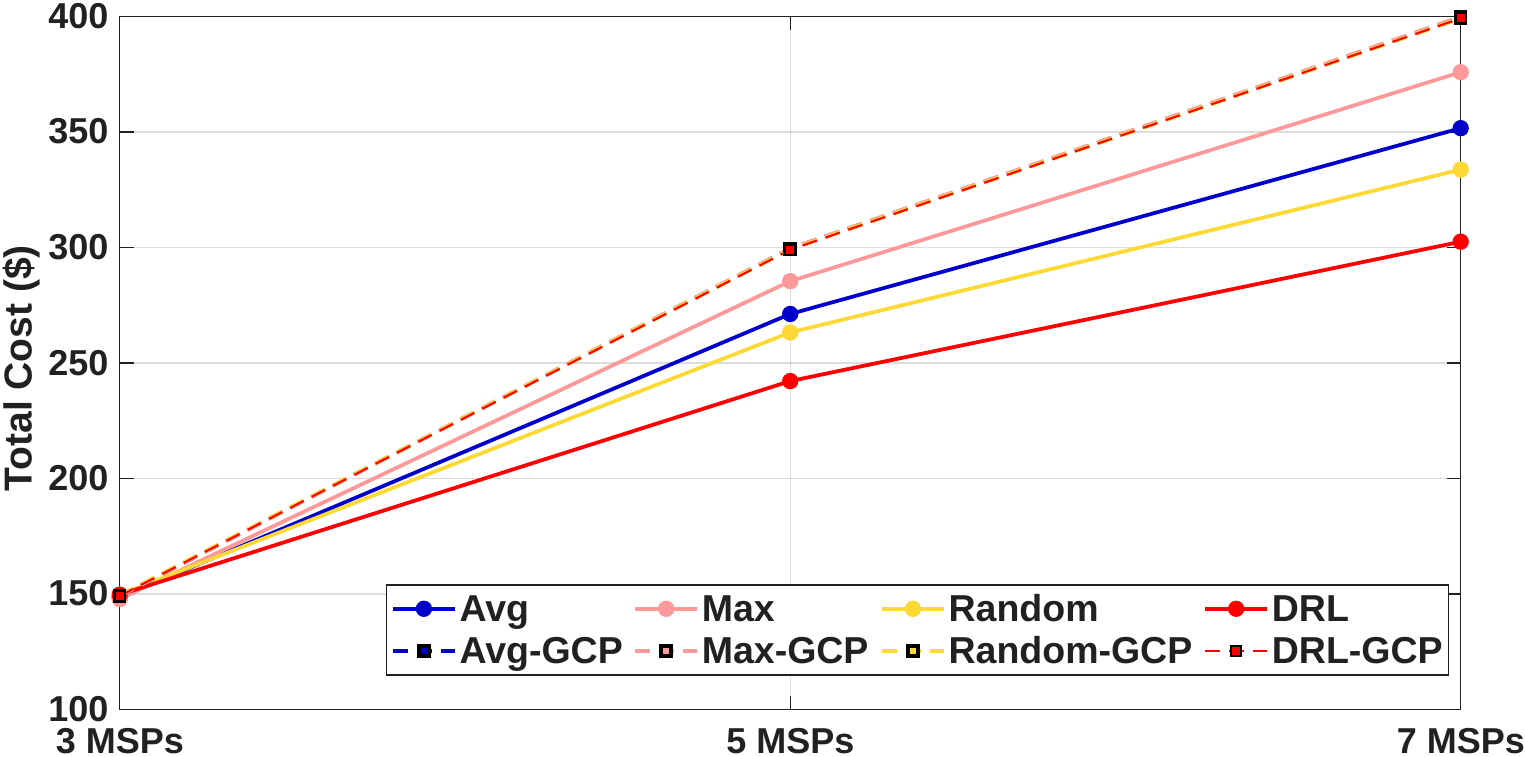} 
        \caption{\updated{200 Timesteps}}
    \end{subfigure}
    \caption{\updated{Comparison of percentage of Total Cost (\$)}}
    \label{fig:cop-cost}
\end{figure*}

The first experiment evaluates the resource allocation performance of the proposed solutions and baseline methods across multiple metrics: the requests completion percentage, indicating total requests done as illustrated in \fref{fig:requests-completion}, the total number of served clients by the requests as shown in \fref{fig:cop-serv}, the requests fulfillment percentage (indicating how many of the done requests satisfy required immersion levels) depicted in \fref{fig:cop-ful}, the total cost presented in \fref{fig:cop-cost}, and fairness measures including the distribution of completed requests among MSPs as detailed in \tref{tab:range} and \tref{tab:gini}.

We remind that in this experiment, each MSP receives multiple requests from their Heads across the time horizon. The requirements of these requests may vary depending on the number of VRCs that need to be served in each Head's VRoom. The MSP can deny/drop any requests if it does not have enough resources and cannot withdraw them.
The requests continue to come into the system, and we measure the performance after 100, 150, and 200 timesteps of these dynamic requests. Additionally, while increasing the number of MSPs adds an extra layer of complexity to the problem, as more MSPs need to collaborate, it can improve the completion ratio of requests.

Starting with request completion results in \fref{fig:requests-completion}, we can observe the different completion percentages of the different baselines across increasing numbers of MSPs and requests as shown in subfigures \fref{fig:requests-completion} (a,b, and c). The main observation of these subgraphs is summarized as follows: firstly, the GCP cooperation mechanism, if not managed correctly, would not increase the requests completion rate; in other words, baselines with GCP can be worse than the non-GCP solution as seen for Max vs Max-GCP in \fref{fig:requests-completion}(a)'s 3 MSP configuration and \fref{fig:requests-completion}(c)'s 7 MSP configuration among others. Secondly, across all subgraphs, we observe the dominance of non-GCP DRL among the non-GCP baselines, and similarly, DRL-GCP dominates among the GCP baselines, resulting in fewer dropped requests from MSPs to Heads.  
Thirdly, although some baseline methods perform similarly to DRL in certain cases, the disparity tends to grow as scenarios become more challenging, particularly with a higher number of requests, as observed in the comparisons of Avg versus DRL and Avg-GCP versus DRL-GCP across all subgraphs. From the same graph, DRL-GCP on average achieves 12-36\% higher request completion than all baselines. 

While \fref{fig:requests-completion} focused on the number of requests to represent performance, \fref{fig:cop-serv} expands this view by highlighting the number of VRCs served by each solution in comparison to the random non-GCP solution. Positive bars indicate improvement, while negative bars reflect inferior performance to non-GCP random. This figure reveals a more pronounced disparity between the solutions, as suboptimal policies may accept many requests that serve fewer VRCs. In contrast, other solutions may accept fewer requests with a greater number of VRCs. Notably, the DRL-GCP solution achieves an impressive 20-60\% increase in the number of clients served on average.

Unlike the previous two figures, \fref{fig:cop-ful} solely focuses on the fulfillment percentage of requests, with comparisons also made against the non-GCP random solution. The significant observation from this figure reveals that the 200-timestep, three-MSP configuration appears to represent the most demanding scenario, as evidenced by the minimal improvement relative to the random solution across all configurations. This occurs because the limited collective budget of the reduced number of MSPs proved insufficient to accommodate the substantial volume of requests. Moreover, in virtually all cases, our DRL solutions outperformed other variants. Additionally, across nearly all configurations, the average policy solution consistently achieved the second-best performance. Overall, our DRL-GCP solution achieved fulfillment rates that were 23-70\% higher than the competing methods, underscoring their effectiveness in managing requests.

\fref{fig:cop-cost} illustrates the total expenses associated with each solution, which is the cost incurred by all MSPs on their Heads. A key takeaway is that among the non-GCP solutions, the option with the lowest monetary cost is the preferred choice. On the other hand, in algorithms that incorporate the GCP, a lower expense may not be deemed optimal, as a higher expense might reflect a greater degree of collaboration among the MSPs, where some MSPs fund others in addition to fulfilling their own needs. In the exact figure, we can see that non-GCP DRL had the lowest cost among the other baselines, while DRL-GCP had a competitive price compared to the other GCP-based baselines.

Finally, in cooperative scenarios, while multiple MSPs work together to increase the number of completed and fulfilled requests, it is crucial to ensure that this cooperation remains balanced and effective. \updated{Specifically, all MSPs should serve requests without one MSP dominating in terms of volume. To measure this, we use two metrics across all solutions and baselines. The first is the range, defined as the difference between the MSP with the most completed requests and the one with the fewest. The second is the Gini coefficient, which measures equality among MSPs and ranges from 0 (perfect equality) to 1 (maximum imbalance, where one MSP serves all requests). The Gini coefficient is calculated over the vector $\mathbf{x}=[x_1,\dots,x_M]$, where $x_m$ is the number of completed requests served by MSP $m$ in the experiment. We use the standard definition: $G(\mathbf{x})=\frac{\sum_{i=1}^{M}\sum_{j=1}^{M}|x_i-x_j|}{2M^2\bar{x}}$, where $\bar{x}=\frac{1}{M}\sum_{i=1}^{M}x_i$.}
. The results in \tref{tab:range} show the range values across all baselines. A smaller value is preferred, highlighted in red bold for non-GCP solutions and in blue bold for GCP-based variants. These results show that our non-GCP DRL performed best among non-cooperative baselines, indicating balanced resource allocation across MSPs. Among GCP variants, our DRL-GCP achieved the lowest range, reaching up to 51\% more fairly distributed requests than baselines. Moreover, the presence of the GCP helped balance the number of served requests among MSPs, as shown by the smaller values in the GCP variants. The same conclusion is supported by \tref{tab:gini} of the gini coefficient, where both DRL and DRL GCP had the most balanced request distributions.

The conclusion from this experiment highlights the superiority of our proposed solutions. They achieved the highest request completion and fulfillment rates, served the largest number of clients, and maintained the lowest cost in the non-GCP setup, while offering a competitive price compared to GCP-based alternatives. Moreover, these solutions demonstrated the most balanced distribution of served requests across MSPs, reinforcing their effectiveness not only in optimizing individual performance metrics but also in promoting fairness and stability in cooperative environments.

\begin{table}
\centering
\caption{Performance Across 100, 150, and 200 Request Configurations (Lower range = better), best results from the non-GCP configuration are displayed in bold red color, while best results from the GCP configuration are displayed in bold blue color}
\scalebox{0.61}{ 
\begin{tabular}{|c|c|c|c|c|c|c|c|c|c|}
\hline
\textbf{Reqs} & \textbf{MSPs} & \textbf{Avg} & \textbf{Max} & \textbf{Random} & \textbf{DRL} & \textbf{Avg-GCP} & \textbf{Max-GCP} & \textbf{Random-GCP} & \textbf{DRL-GCP} \\
\hline
\multirow{3}{*}{100} 
& 3 MSPs & 58 & 76 & 62 & \textcolor{red}{\textbf{42}} & \textcolor{blue}{\textbf{0}} & 50 & 19 & \textcolor{blue}{\textbf{0}} \\
& 5 MSPs & 58 & 76 & 63 & \textcolor{red}{\textbf{42}} & \textcolor{blue}{\textbf{0}} & 17 & \textcolor{blue}{\textbf{0}} & \textcolor{blue}{\textbf{0}} \\
& 7 MSPs & 63 & 79 & 67 & \textcolor{red}{\textbf{49}} & \textcolor{blue}{\textbf{0}} & 38 & \textcolor{blue}{\textbf{0}} & \textcolor{blue}{\textbf{0}} \\
\hline
\multirow{3}{*}{150} 
& 3 MSPs & 92 & \textcolor{red}{\textbf{78}} & 113 & 92 & 70 & 50 & 69 & \textcolor{blue}{\textbf{46}} \\
& 5 MSPs & 120 & 126 & 114 & \textcolor{red}{\textbf{96}} & 10 & 88 & 43 & \textcolor{blue}{\textbf{0}} \\
& 7 MSPs & 124 & 129 & 119 & \textcolor{red}{\textbf{101}} & 40 & 85 & 64 & \textcolor{blue}{\textbf{0}} \\
\hline
\multirow{3}{*}{200} 
& 3 MSPs & 92 & \textcolor{red}{\textbf{78}} & 112 & 106 & \textcolor{blue}{\textbf{50}} & \textcolor{blue}{\textbf{50}} & 69 & 88 \\
& 5 MSPs & 170 & 176 & 163 & \textcolor{red}{\textbf{142}} & 97 & 93 & 113 & \textcolor{blue}{\textbf{43}} \\
& 7 MSPs & 174 & 179 & 167 & \textcolor{red}{\textbf{149}} & 85 & 85 & 120 & \textcolor{blue}{\textbf{78}} \\
\hline
\end{tabular}
}
\label{tab:range}
\end{table}
\begin{table}
\centering
\caption{Gini-Coefficient Across 100, 150, and 200 Request Configurations (Lower value = better)}
\scalebox{0.61}{
\begin{tabular}{|c|c|c|c|c|c|c|c|c|c|}
\hline
\textbf{Reqs} & \textbf{MSPs} & \textbf{Avg} & \textbf{Max} & \textbf{Random} & \textbf{DRL} & \textbf{Avg-GCP} & \textbf{Max-GCP} & \textbf{Random-GCP} & \textbf{DRL-GCP} \\
\hline
\multirow{3}{*}{100} 
& 3 MSPs & 0.16 & 0.25 & 0.17 & \textbf{\textcolor{red}{0.11}} & \textbf{\textcolor{blue}{0.00}} & 0.17 & 0.05 & \textbf{\textcolor{blue}{0.00}} \\
& 5 MSPs & 0.11 & 0.18 & 0.12 & \textbf{\textcolor{red}{0.07}} & \textbf{\textcolor{blue}{0.00}} & 0.04 & \textbf{\textcolor{blue}{0.00}} & \textbf{\textcolor{blue}{0.00}} \\
& 7 MSPs & 0.15 & 0.23 & 0.16 & \textbf{\textcolor{red}{0.11}} & \textbf{\textcolor{blue}{0.00}} & 0.09 & \textbf{\textcolor{blue}{0.00}} & \textbf{\textcolor{blue}{0.00}} \\
\hline
\multirow{3}{*}{150} 
& 3 MSPs & 0.26 & 0.26 & 0.26 & \textbf{\textcolor{red}{0.18}} & 0.15 & 0.17 & 0.17 & \textbf{\textcolor{blue}{0.09}} \\
& 5 MSPs & 0.22 & 0.25 & 0.19 & \textbf{\textcolor{red}{0.12}} & 0.02 & 0.19 & 0.08 & \textbf{\textcolor{blue}{0.00}} \\
& 7 MSPs & 0.27 & 0.30 & 0.23 & \textbf{\textcolor{red}{0.18}} & 0.06 & 0.17 & 0.11 & \textbf{\textcolor{blue}{0.00}} \\
\hline
\multirow{3}{*}{200} 
& 3 MSPs & 0.26 & 0.26 & 0.26 & \textbf{\textcolor{red}{0.20}} & \textbf{\textcolor{blue}{0.17}} & 0.17 & 0.17 & 0.17 \\
& 5 MSPs & 0.27 & 0.30 & 0.23 & \textbf{\textcolor{red}{0.18}} & 0.15 & 0.20 & 0.18 & \textbf{\textcolor{blue}{0.06}} \\
& 7 MSPs & 0.33 & 0.36 & 0.29 & \textbf{\textcolor{red}{0.22}} & 0.15 & 0.17 & 0.17 & \textbf{\textcolor{blue}{0.09}} \\
\hline
\end{tabular}
}
\label{tab:gini}
\end{table}

\subsubsection{Experiment 2: Adaptability to changes in environment}

\begin{figure}
    \centering
    \begin{subfigure}[b]{0.32\textwidth}
        \centering
        \includegraphics[width=\textwidth]{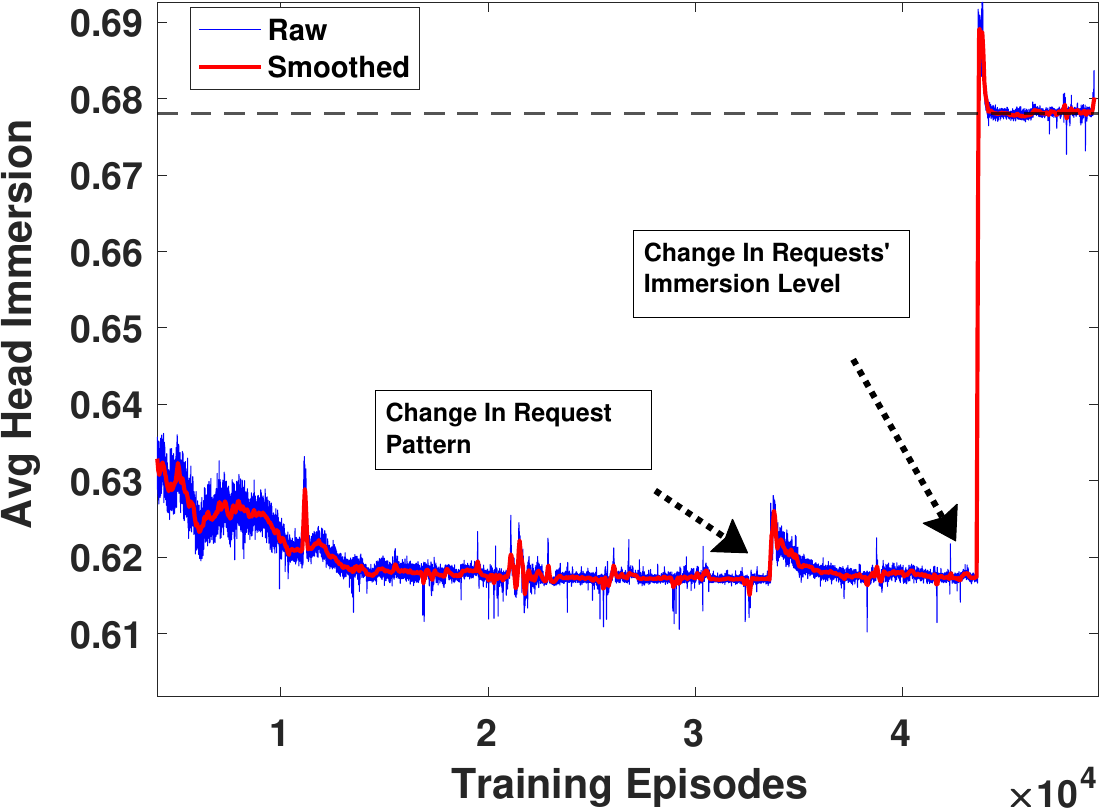} 
        \caption{Average head Served adaptation}
        \label{fig:cop-avghead-done-conv}
    \end{subfigure}
     \begin{subfigure}[b]{0.32\textwidth}
        \centering
        \includegraphics[width=\textwidth]{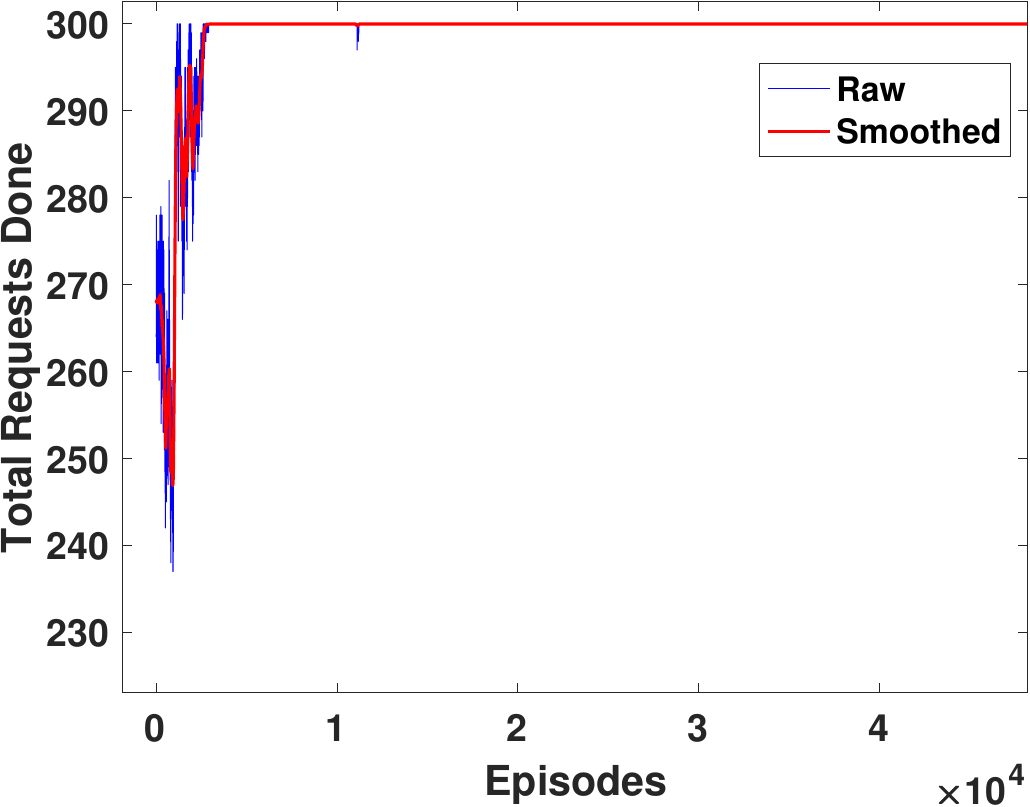} 
        \caption{Requests convergence adaptation }
        \label{fig:cop-req-done-conv}
    \end{subfigure}
    \begin{subfigure}[b]{0.32\textwidth}
        \centering
        \includegraphics[width=\textwidth]{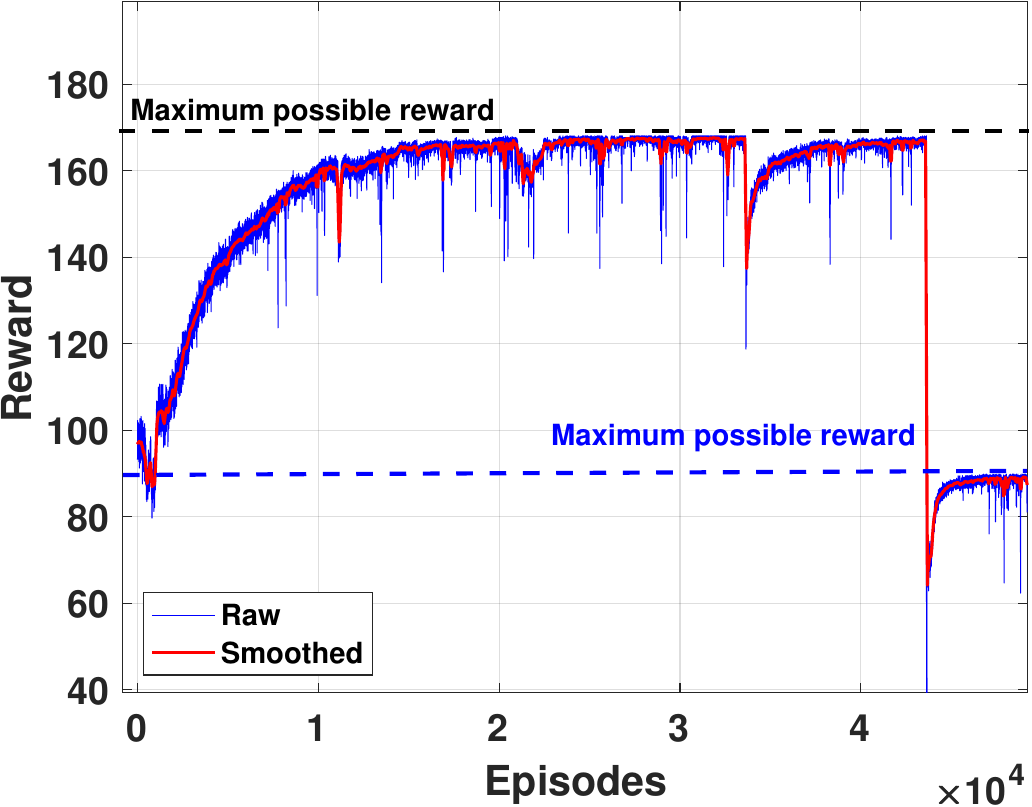} 
        \caption{Reward re-convergence}
        \label{fig:cop-rew-done-conv}
    \end{subfigure}
    \begin{subfigure}[b]{0.32\textwidth}
        \centering
        \includegraphics[width=\textwidth]{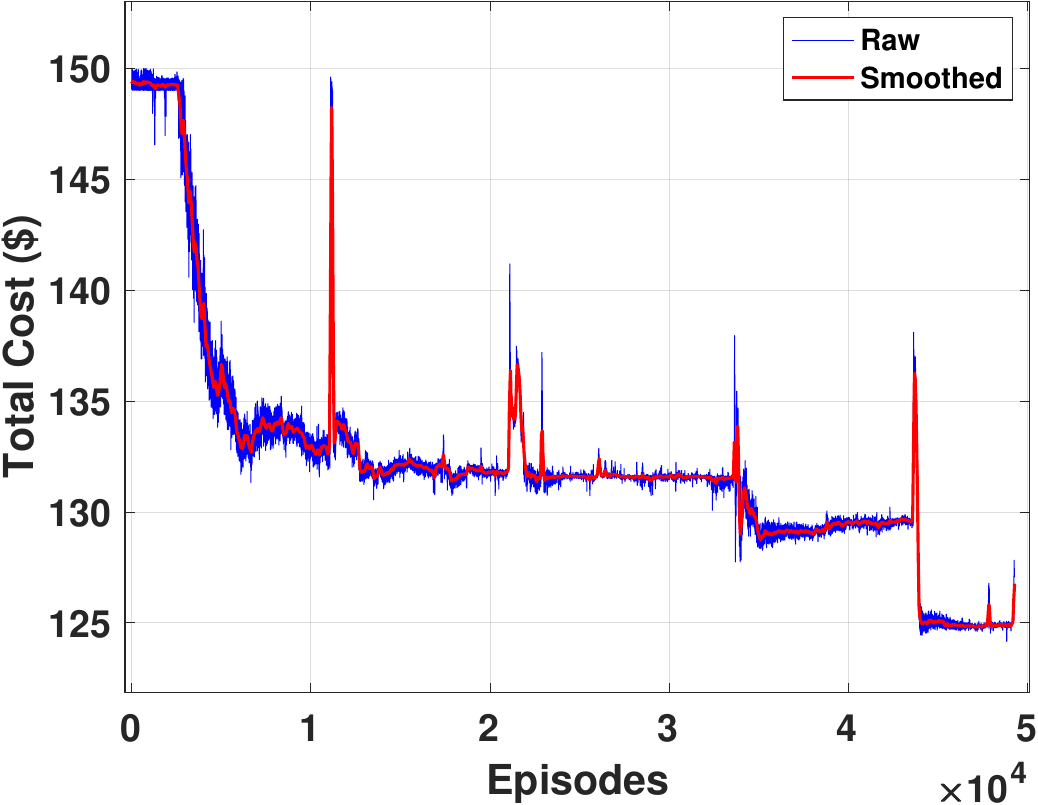}
        \caption{Cost convergence}
        \label{fig:cop-cost-done-conv}
    \end{subfigure}
   
    \caption{Assessments of the Cooperative control of the MSPs}
    \label{fig:non-cooperative_conv-adaptaton}
\end{figure}
\updated{Moreover, since the Metaverse environment is dynamic, with various events emerging suddenly and diverse needs for Heads, MSPs must be prepared to handle changes in request patterns from the Heads. Therefore, in the second experiment, we evaluate the resilience and adaptability of our solutions by subjecting the environment to multiple unexpected changes and observing how our DRL agent handles the resource allocation for these requests. In our case, we consider two types of changes.
The first is a sudden shift in the request distribution of the Heads, simulating unexpected requests, such as seasonal concerts. Such a change can negatively impact poorly trained DRL agents and any history-based baselines, potentially leading to under- or over-provisioning.
The second change simulates a dynamic shift in the immersion threshold required by some of the Heads. This reflects real-world scenarios in which certain Heads may demand higher-quality experiences in their VRooms. In that experiment, the DRL model was initially trained to meet an immersion threshold of 60\%, which was later increased to 67\%.

The results, shown in \fref{fig:non-cooperative_conv-adaptaton}, highlight the system's adaptability. Specifically, \fref{fig:non-cooperative_conv-adaptaton}(a) shows that our DRL-GCP agent quickly adapted to the two environment changes, shown by the average Head immersion after only a few episodes. \fref{fig:non-cooperative_conv-adaptaton}(b) indicates that our agent did not drop any requests even after the environment changed. \fref{fig:non-cooperative_conv-adaptaton}(c) illustrates the agent’s reward re-convergence, while \fref{fig:non-cooperative_conv-adaptaton}(d) presents the total cost changes. 
The sudden drops in \fref{fig:non-cooperative_conv-adaptaton}(c)(d) occur immediately after the injected distribution/threshold shifts. At that moment, the previously learned policy is temporarily mismatched to the new demand/constraint regime (e.g., higher target immersion or different request mix), which reduces reward and alters cost until the policy re-stabilizes under the new regime; the subsequent recovery shows the agent re-adapting over the next episodes.

}

\subsubsection{Experiment 3: Optimal Withdrawal-to-Insertion Ratio Analysis}
Finally, since systems employing shared pools are subject to the free-rider problem, in which certain MSPs with suboptimal resource-allocation strategies exploit the collective resource by continuously withdrawing from the GCP without making any effective contribution, we address this issue in two steps. The first step is to enhance the resource allocation policies of all MSPs through DRL, and the second is to implement withdrawal constraints proportional to insertion limits. However, given the inherent complexity in determining optimal withdrawal thresholds, we conducted a third experiment specifically designed to evaluate the withdrawal-to-insertion ratio that maximizes system-wide performance efficiency. \updated{In our third experiment, shown in \fref{fig:cop-sp_limit}, we examined multiple withdrawal-to-insertion ratios. For instance, a factor of 1.6 indicates that an MSP can receive up to 160\% of the GC units it has inserted into the GCP when needed. From the same figure, it is evident that a higher withdrawal limit increases the number of requests and served clients while reducing the number of failed clients. This also suggests that our centralized DRL solution does not explicitly require the setup of a safeguard to limit free-riders, the DRL agent automatically optimizes the allocation to balance between varying types and budget MSPs for the long-term goal of increasing the number of overall requests.}







\begin{figure}
    \centering
    \includegraphics[width=1\linewidth]{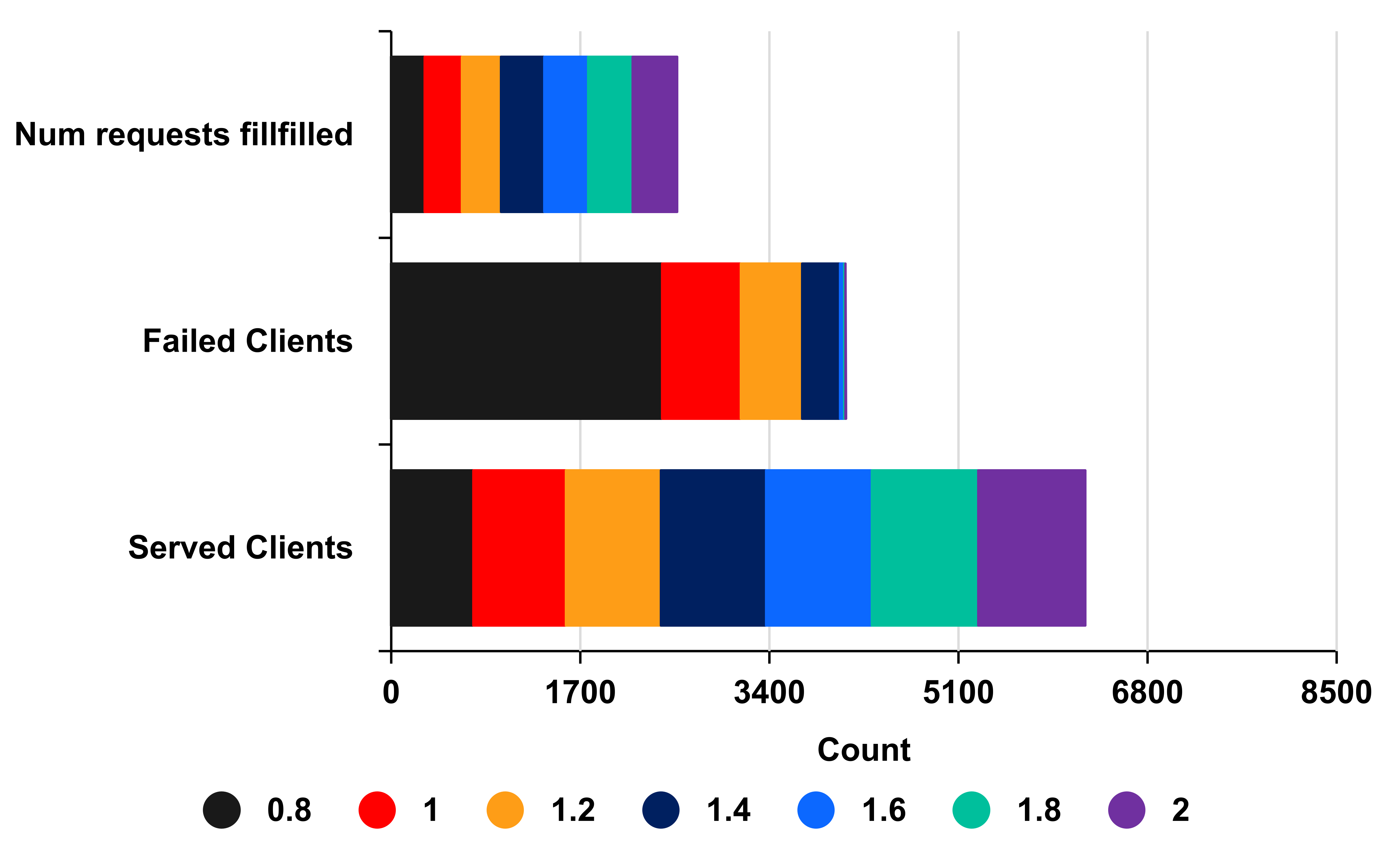}
    \caption{\updated{The effect of limiting the withdrawal-to-insertion.}}
    \label{fig:cop-sp_limit}
\end{figure}

\section{CONCLUSION AND FUTURE WORK}
\label{sec:conc_civic}
This paper introduced CIVIC, a novel framework for optimizing resource allocation and enabling intelligent cooperation across multiple MSPs. Unlike prior works that overlooked digital twin accuracy in immersion modeling, relied on rigid resource allocation strategies, or were constrained to single-MSP settings, CIVIC provides a unified solution that integrates virtual environment rendering, digital twin synchronization, and immersion-aware provisioning across diverse virtual rooms in dynamic multi-MSP ecosystems.

We formulated the resource allocation challenge as two NP-hard problems in different settings: a non-cooperative model where MSPs act independently, and a cooperative model that leverages GCP for dynamic inter-MSP resource sharing. To address the computational complexity, we designed DRL agents that account for system heterogeneity and load uncertainty while coordinating resource distribution. Extensive experiments demonstrated that our DRL-based approaches consistently outperformed baseline strategies in both efficiency and fairness.

Specifically, our DRL-GCP cooperative model achieved improvements of 12–36\% in request completion, 23–70\% in fulfillment rate, 20–60\% in client capacity, and up to 51\% in equitable request distribution—all while maintaining competitive operational costs. These results validate that CIVIC delivers robust, adaptive, and scalable performance under fluctuating load conditions, establishing its viability for real-world multi-MSP Metaverse deployments powering the next generation of immersive digital experiences.

\bibliographystyle{IEEEtran}
\bibliography{references}

\appendix{}
\label{sec:appndx}
\begin{theorem}[NP-hardness of the cooperative problem $\bar P$]\label{thm:nphard-barP}
{The decision version of $\boldsymbol{\bar{P}}$ (objective~\eqref{eq:objective2} with constraints (c1)--(c13) and GCP update~\eqref{eq:pool_dynamics}) is NP-hard.}
\end{theorem}

\begin{proof}[Proof]
To show the claimed hardness in a transparent way, we reduce from 0--1 \emph{Knapsack}. Given items $i=1,\dots,n$ with weights $w_i$, values $v_i\in\mathbb{Z}_{+}$, capacity $W$, and target value $V^\star$, construct an instance of $\boldsymbol{\bar{P}}$ with $M=n$ MSPs and a short horizon $T=2$.

\textbf{Step 1 (calibration).} For each MSP $i$, fix the VRoom and model parameters so that meeting the immersion threshold for its (single) active Head requires a \emph{minimal} total allocation cost exactly $w_i$. This uses the same immersion and cost functions already in $\bar{P}$; feasibility checking for a proposed allocation remains polynomial.

{\textbf{Step 2 (GCP capacity).} Set all per-MSP budgets at $t=2$ to zero, by allowing a subset of MSPs at time $t=1$ to donate their full surpluses so that the GCP holds exactly $W$ credits at $t=2$. This is permitted by the donation variable $\delta_m^{(t)}\!\in[0,1]$, the donation rule (c10), and the GCP recursion \eqref{eq:pool_dynamics}, which carries donations from $t=1$ into available credits at $t=2$.}

{\textbf{Step 3 (equivalence).} At $t=2$, constraint (c11) limits the aggregate MSP deficit of the selected fully served MSPs by the available $W$ credits in the GCP. Because $L_m^{(t)}=1$ iff \emph{all} active Heads of MSP $m$ meet the threshold (c3), choosing which MSPs to fully serve at $t=2$ consumes exactly $\sum_{i\in S} w_i$ credits and yields total value $\sum_{i\in S} v_i$ in the objective sum $\sum_{m} L_m^{(2)}$. Hence, there exists a feasible allocation with $\sum_m L_m^{(2)} \ge V^\star$ iff the knapsack instance admits a subset $S$ with $\sum_{i\in S} w_i \le W$ and $\sum_{i\in S} v_i \ge V^\star$. Therefore, $\bar P$-DEC is NP-hard.}
\end{proof}

\begin{theorem}[NP-hardness of the non cooperative problem $\boldsymbol{P}$]\label{thm:nphard-P}
{The decision version of $\boldsymbol{P}$ (objective~\eqref{eq:objective} with constraints (c1)--(c8)) is NP-hard.}
\end{theorem}

\begin{proof}[Proof]
We again reduce from 0--1 \emph{Knapsack}. To isolate the source of combinatorial hardness, consider a single-MSP instance ($M{=}1$) with horizon $T$ and exactly one active Head at each $t\in\{1,\dots,T\}$. Calibrate parameters so that the \emph{minimal} cost to reach the immersion threshold at time $t$ equals a prescribed weight $w_t$.

Use the paper’s time-coupled budget semantics: allocation at $t$ reduces the remaining budget available at $t{+}1$. Set the initial budget to $B$ and interpret the per-timestep budget bound (c4) as the remaining capacity at each step. Then selecting the set of timesteps $S\subseteq\{1,\dots,T\}$ in which we fully serve the Head corresponds to picking items with total “weight” $\sum_{t\in S} w_t$ not exceeding $B$, and yields objective value $\sum_{t\in S} L_1^{(t)} = |S|$. Thus, deciding whether there exists a feasible allocation with $\sum_{t=1}^T L_1^{(t)} \ge K$ and total spend $\sum_{t\in S} w_t \le B$ is exactly the 0--1 knapsack decision problem. Hence $\boldsymbol{P}$-DEC is NP-hard.
\end{proof}

\textbf{Remark:}
{These reductions use only the all-or-nothing fulfillment rule (c3), the local budget constraints (c4), and, for $\boldsymbol{\bar{P}}$, the shared-pool coupling through (c10)--(c13) and the GCP update \eqref{eq:pool_dynamics}. They already show combinatorial hardness in simplified settings. In the full models, the objective and constraints additionally include non-convex functions (SSIM with a negative exponent in bitrate, the resource-scaling Hessian with negative mixed second derivatives, and the harmonic-mean DT accuracy), which the paper proves via Hessian analysis; thus both $\boldsymbol{P}$ and $\boldsymbol{\bar{P}}$ are non-convex MINLPs, further precluding polynomial-time global optimization and motivating DRL-based solution methods.}
\end{document}